\documentclass[11pt]{article}
\usepackage{amsmath}
\usepackage{graphicx,psfrag,epsf}
\usepackage{enumerate}
\usepackage{psfrag,epsf}
\usepackage{url} 
\usepackage{amsmath,amssymb,amsfonts,amsthm,mathtools,mathrsfs}
\usepackage{booktabs}
\usepackage{float}
\usepackage{subcaption}
\usepackage[title]{appendix}

\usepackage{bm,bbm}
\usepackage{color,float}
\usepackage{algorithm,algpseudocode}
\usepackage[symbol]{footmisc}
\usepackage{scalefnt}
\usepackage{authblk} 
\usepackage{multirow,centernot}
\usepackage{natbib}
\usepackage[colorlinks=true, citecolor=blue, urlcolor=blue]{hyperref}
\usepackage{dsfont}
\usepackage[title]{appendix}
\usepackage{newtxtext,newtxmath} 
\usepackage{mhchem} 

\usepackage[left=1in,top=1.1in,right=0.5in,bottom=1in]{geometry}

\theoremstyle{definition}

\newtheorem{theorem}{Theorem}[section]

\newtheorem{proposition}[theorem]{Proposition}
\newtheorem{remark}[theorem]{Remark}

\makeatletter
\def\@seccntformat#1{\@ifundefined{#1@cntformat}%
	{\csname the#1\endcsname\quad}
	{\csname #1@cntformat\endcsname}
}
\makeatother

\newif\ifShowComments
\ShowCommentstrue
\def\strutdepth{\dp\strutbox}
\def\druk#1{\strut\vadjust{\kern-\strutdepth
        {\vtop to \strutdepth{%
                \baselineskip\strutdepth\vss
                        \llap{\hbox{#1}\quad}\null}}}}

\title{\bf
On classes of distributions on the unit interval: structural properties and application to inequality data}

\author{
\text{Roberto Vila}$^{1}\thanks{Corresponding author: Roberto Vila, email: {rovig161@gmail.com}
}$,
\,
\text{Helton Saulo}$^{1,2}$,
\, 
\text{Poliana Matos}$^{1}$
 and 
\text{Subhankar Dutta}$^{3}$
\\
{\small $^{1}$ Department of Statistics, University of Brasilia, Brasilia, 70910-900, Brazil}\\
{\small $^{2}$ Department of Economics, Federal University of Pelotas, Pelotas, 96010-610, Brazil}\\
{\small $^{3}$Department of Mathematics, Bioinformatics and Computer Applications, Maulana Azad National Institute of Technology, Bhopal, 462003, India}
}
\begin{document}
	\maketitle 	
	\begin{abstract}
%
%
Probability distributions defined on the unit interval are widely used in fields ranging from econometrics to reliability studies. Traditional models such as the beta and Kumaraswamy distributions are well-established due to their flexibility and tractability. In this paper, we introduce two novel families of unit-interval distributions derived via non-injective transformations of the gamma ratio. These transformations, denoted $S_r$ and $T_r$, allow the construction of new random variables with support on $(0,1)$ and admit simple closed-form expressions for their densities when the underlying variables are independent gamma distributed. Notably, for $r = 1/2$, these constructions yield sample-based estimators of the Gini and Atkinson indices, establishing a direct link with classical inequality measures. We derive the distributional laws, cumulative distribution functions, quantile functions, and raw moments, and discuss maximum likelihood estimation for the proposed models. A Monte Carlo simulation study is conducted to assess the finite sample behavior of the maximum likelihood estimators under different parameter configurations. An application to cross-country Gini index data illustrates the flexibility and practical relevance of the proposed distributions in modeling real inequality indicators.



	\end{abstract}
	\\
	{\small{\bfseries Mathematics Subject Classification (2010).} {MSC 60E05 $\cdot$ MSC 62Exx $\cdot$ MSC 62Fxx.}}
{
	\hypersetup{hidelinks}
}

\section{Introduction}

Probability models with support on the unit interval are useful in many fields, including econometrics \citep{Cribari2010}, natural sciences \cite{geissingeretal:2022}, and reliability \citep{crossetal:2006}, among others. In the context of economic analysis, several key indicators are naturally constrained to $(0,1)$, most notably measures of income inequality such as the Gini index \citep{Gini1936}. This coefficient is routinely used to summarize the degree of income concentration across countries. 

Traditionally, the beta \citep{fcn:04} and Kumaraswamy \citep{kuma1980} distributions have been the workhorses for modeling bounded continuous data, due to their flexibility and tractability. More recently, alternative unit--interval families have emerged, see, for example, \citet{Vila2024b,Vila03072024}, to address features such as bimodality and complex shapes. Despite these advances, most existing models are essentially phenomenological, offering limited theoretical connection with the stochastic mechanisms underlying classical inequality measures.

The beta distribution itself admits an interpretation as the ratio of two positive random variables. Indeed, when $X,Y \sim \mathrm{Gamma}(\alpha_i,\lambda)$ are independent and identically scaled,
\[
\frac{X}{X+Y} \sim \mathrm{Beta}(\alpha_1,\alpha_2),
\]
yielding the familiar beta model. However, this construction does not explicitly exploit the intrinsic relationship between income distributions and inequality indices. In particular, the Gini coefficient can be expressed in terms of ratios of independent gamma variables, which suggests that suitable transformations of $X/(X+Y)$ may lead to probability models that are naturally tailored for inequality data.

In this work, we introduce two new families of unit--interval distributions defined by non--injective transformations of the ratio $X/(X+Y)$. Writing
\[
S_r(u) = [\,1 - 4u(1-u)\,]^r,
\qquad
T_r(u) = 1 - [\,4u(1-u)\,]^r,
\quad 0 < u < 1,\ r > 0,
\]
and setting $W = S_r\bigl(X/(X+Y)\bigr)$ and $Z = T_r\bigl(X/(X+Y)\bigr)$, one obtains random variables supported on $(0,1)$ whose densities admit simple closed--form expressions when $X$ and $Y$ are independent $\mathrm{Gamma}(\alpha_i,\lambda_i)$. The particular case $r = 1/2$ recovers sample--based estimators of the Gini and Atkinson indices from a pair $(X,Y)$, thereby establishing a direct link between the proposed families and the classical measures of inequality \citep{Gini1936,Atkinson1970}. This connection provides a strong motivation for the use of the proposed models in the statistical analysis of Gini data.

The remainder of this paper is organized as follows. In Section~\ref{sec:02}, we introduce the new unit--interval families via the non--injective transformations $S_r$ and $T_r$ of the gamma ratio. In Section~\ref{sec:03}, we derive their distributional laws in terms of the underlying $X/(X+Y)$ density. In Section~\ref{sec:04}, we obtain closed--form expressions for the density and cumulative distribution functions of $W$ and $Z$. In Section~\ref{sec:05}, we present the quantile functions. In Section~\ref{sec:06}, we develop closed--form formulae for selected raw moments under special parameter regimes. In Section~\ref{sec:07}, we formulate maximum likelihood estimation, derive the likelihood equations, and discuss numerical solution strategies. In Section~\ref{simulation}, we report the results of a Monte Carlo study. In Section~\ref{Applications}, we illustrate the practical performance of the proposed models by fitting them to real Gini index data. Finally, in Section~\ref{sec:08}, we provide some concluding remarks. Additional technical derivations are collected in the Appendix.

\section{The unit probabilistic models}\label{sec:02}

Let $X$ and $Y$ be two (absolutely) continuous random variables, with positive support, both defined in the same probability space $(\Omega,\mathscr{B}(\Omega),\mathbb{P})$. Here, $\mathscr{B}(\Omega)$ denotes the Borel $\sigma$-algebra on $\Omega$ and $\mathbb{P}:\mathscr{B}(\Omega)\to [0,1]$ is a probability measure. 
In this work we define news random variables as the composition of a surjective transformation $\varphi_r:(0,1)\to (0,1)$ with the ratio $X/(X+Y)$, that is, 
\begin{align}\label{rep-stochastic}
	\varphi_r\circ \left({X\over X+Y}\right)
	=
	\varphi_r\left({X\over X+Y}\right):\Omega\to (0,1),
\end{align}
so that the addition of the parameter $r>0$ generates flexibility in the proposed models. 

Note that the complexity of the probabilistic models defined in \eqref{rep-stochastic} depends on the choices of variables $X$ and $Y$ and on its correlation structure, which are innumerable. 
{As an illustration, reference \citep{colin2025} explores various probabilistic aspects of model \eqref{rep-stochastic} under the conditions $\varphi_r(x)=[1-4x(1-x)]^r$ and $r = 1/2$, assuming that the variables $X$ and $Y$ are correlated Gamma-distributed with the same shape and scale parameters, and a linear correlation coefficient $\rho$.}
 For simplicity of presentation, from now on we choose  $X$ and $Y$ independent {(not necessarily identically distributed)} with gamma distributions, denoted by $X\sim\text{Gamma}(\alpha_1,\lambda_1)$, $\alpha_1,\lambda_1>0$, and $Y\sim\text{Gamma}(\alpha_2,\lambda_2)$, $\alpha_2,\lambda_2>0$, respectively. That is, the density functions of $X$ and $Y$ are given by
\begin{align*}
f_X(x)={\lambda_1^{\alpha_1}\over\Gamma(\alpha_1)}\, x^{\alpha_1-1}\exp(-\lambda_1 x), 
\quad 
f_Y(y)={\lambda_2^{\alpha_2}\over\Gamma(\alpha_2)}\, y^{\alpha_2-1}\exp(-\lambda_2 y), \quad x,y>0,
\end{align*}
respectively, where $\Gamma (z)=\int _{0}^{\infty }t^{z-1}\exp(-t){\text{d}}t$, $z>0,$ is the (complete) gamma function.
In this direction, a notable result in distribution theory, widely recognized in the literature, occurs when $\lambda_1=\lambda_2$ and $\varphi_r(x)=x,$ $0<x<1$, as this leads to 
$
	{X/(X+Y)}\sim\text{Beta}(\alpha_1,\alpha_2).
$
Alternative choices of $X$ and $Y$ beyond the gamma distributions have been thoroughly explored in the literature even for the case $\varphi_r(x)=x,$ $0<x<1$. For instance, \cite{Vila2024a} selected $X$ and $Y$ as independent random variables with general distributions, whereas \cite{Vila2024b} examined $X$ and $Y$ with a bivariate unit-logsymmetric distribution. Additionally, 
\cite{Vila2025} chose $X$ and $Y$ correlated according to the bivariate Birbaum-Saunders (BS) distribution, and \cite{Vila-Quintino2025} considered $X$ and $Y$ following a bivariate Fréchet distribution.

Simple examples include choosing the transformation $\varphi_r$ as injector functions, such as
$\varphi_r(x)
	=x^r,
$
$
	\varphi_r(x)
	=
	{2^rx^r/(x+1)^r},
	$
$
	\varphi_r(x)
	=
	{(1-x)^r/(x+1)^r},
$
$0<x<1$; $r>0$.
%
The scope of this work centers on non-injective  transformations $\varphi_r$, like:
\begin{align}\label{transf-using}
	S_{r}(x)
	\equiv
	\varphi_r(x)
	=
	[1-4x(1-x)]^r,
	\quad 
	T_{r}(x)
	\equiv
	\varphi_r(x)
	=
	1-[4x(1-x)]^r,
		\quad
	0<x<1; \ r>0.
\end{align}
Note that both transformations above are symmetric and have a zero at the point $x=1/2$ and assume the value of $1$ at the points $x=0$ and $x=1$.

Using the transformations $S_{r}$ and $T_{r}$ in \eqref{transf-using}, simple algebraic manipulations show that, the new random variables $W$ and $Z$, defined as
\begin{align}
W
\equiv 
	S_{r}\left({X\over X+Y}\right)
	=
	\left({Y-X\over X+Y}\right)^{2r},
\quad
Z
\equiv
T_{r}\left({X\over X+Y}\right)
=
1-
{(XY)^r\over \displaystyle \left[{1\over 2}(X+Y)\right]^{2r}}, \label{main-rv-1}
\end{align}
respectively,
have support in the unit interval.
\begin{remark}
	Note that when $r=1$, we get $W=Z$.
\end{remark}
\begin{remark}
The value of $r=1/2$ in formulas \eqref{main-rv-1} is special, because for this value the variables $W$ and $Z$ take the forms
\begin{align*}
	&W
	=
	{\vert X-Y\vert\over X+Y}
\quad \text{and} \quad 
	Z
	=
	1-
	{\sqrt{XY}\over {1\over 2}(X+Y)},
\end{align*}
which coincide, respectively, with the Gini coefficient \citep{Gini1936} and Atkinson index \citep{Atkinson1970} estimators formed by the random sample $(X,Y)$ of size $n=2$.
\end{remark}


\section{Deriving the laws of $W$ and $Z$}\label{sec:03}

Note that, for any Borel set $B$ on $\mathbb{R}$, we have
\begin{align}\label{law-w}
	\mathbb{P}\left(W\in B\right)
	=
	\mathbb{P}\left(S_{r}\left({X\over X+Y}\right)\in B\right)
	&=
	\mathbb{P}\left(\left\{\omega\in\Omega:{X(\omega)\over X(\omega)+Y(\omega)}\in S_{r}^{-1}(B)\right\}\right)
	=
	\int_{S_{r}^{-1}(B)} f_{X\over X+Y}(u){\rm d}u,
\end{align}
where $S_{r}^{-1}(B)=\{x\in\mathbb{R}:S_r(x)\in B\}$ is the inverse image of set $B$ under the function $S_r$. Since the Borel $\sigma$-algebra on $\mathbb{R}$ is generated by open or closed sets, without loss of generality, assume that $B=(a,b)$, with $a<b$. Hence,
\begin{align*}
	S_{r}^{-1}(B)
	=
	\begin{cases}
		\left({1-b_*^{^{1\over 2r}}\over 2},{1-a_*^{^{1\over 2r}}\over 2}\right)\bigcup \left({1+a_*^{^{1\over 2r}}\over 2},{1+b_*^{^{1\over 2r}}\over 2}\right), & B\cap (0,1)\neq \emptyset,
		\\[0,3cm]
		\emptyset, & B\cap (0,1)= \emptyset,	
	\end{cases}
\end{align*}
where $a_*=a\vee 0=\max\{a,0\}$ and $b_*=b\wedge 1=\min\{b,1\}$. Therefore, by using \eqref{law-w},  the law of $W$ can be written as
\begin{align}\label{law-W}
	\mathbb{P}\left(W\in B\right)
	=
	\begin{cases}
		\displaystyle
	\int_{{\left({1-b_*^{^{1\over 2r}}\over 2},{1-a_*^{^{1\over 2r}}\over 2}\right)\bigcup \left({1+a_*^{^{1\over 2r}}\over 2},{1+b_*^{^{1\over 2r}}\over 2}\right)}} f_{X\over X+Y}(u){\rm d}u, & B\cap (0,1)\neq \emptyset,
			\\[0,3cm]
	0, & B\cap (0,1)= \emptyset.
		\end{cases}
\end{align}

Similarly, we can show that the law of $Z$ can be expressed as
{\small 
\begin{align}\label{law-Z}
	\mathbb{P}\left(Z\in B\right)
	=
	\begin{cases}
			\displaystyle
		\int_{{\left({1-\sqrt{1-(1-b_*)^{1\over r}}\over 2},{1-\sqrt{1-(1-a_*)^{1\over r}}\over 2}\right)
			\bigcup 
			\left({1+\sqrt{1-(1-a_*)^{1\over r}}\over 2},{1+\sqrt{1-(1-b_*)^{1\over r}}\over 2}\right)}} f_{X\over X+Y}(u){\rm d}u, & B\cap (0,1)\neq \emptyset,
			\\[0,3cm]
			0, & B\cap (0,1)= \emptyset.
				\end{cases}
\end{align}
}

By \eqref{law-W} and \eqref{law-Z}, we note that  knowledge of the laws of $W$ and $Z$ relies on knowledge of  density of the independent ratio $X/(X+Y)$, with $X\sim\text{Gamma}(\alpha_1,\lambda_1)$ and $Y\sim\text{Gamma}(\alpha_2,\lambda_2)$, which is well-known in literature \citep[see Example 6 of reference][]{Vila2024a}:
\begin{align}\label{formula-pdf-gammas}
	f_{X\over X+Y}(v)
	=
	{1\over\lambda_1^{\alpha_2}\lambda_2^{\alpha_1} \text{B}(\alpha_1,\alpha_2)} \, 
	\dfrac{v^{\alpha_1-1}(1-v)^{\alpha_2-1}}
	{\displaystyle
	\left[\left({1\over\lambda_2}-{1\over\lambda_1}\right) v	+
	{1\over\lambda_1}\right]^{\alpha_1+\alpha_2}},
	\quad 0<v<1,
\end{align}
where $\text{B}(z_{1},z_{2})
=\int_{0}^{1}t^{z_{1}-1}(1-t)^{z_{2}-1}{\rm d}t$ is the (complete) beta function.


\section{Distribution and density functions of $W$ and $Z$}\label{sec:04}

By taking $a\to-\infty$ and $b=w\in(0,1)$ in \eqref{law-W}, we have $a_*=0$ and $b_*=w$. Then the  cumulative distribution function (CDF) of $W$ can be written as  
\begin{align}\label{exp-f-w}
	F_W(w)
	=
	\int_{1-w_r}^{w_r} f_{X\over X+Y}(u){\rm d}u,
		\quad 
	0<w<1,
\end{align}
where we have adopted the notation
\begin{align}\label{def-w-r}
	w_r\equiv {1+w^{\!^{1\over 2r}}\over 2}.
\end{align}
Consequently, the probability density function (PDF) of $W$ is given by (for 	$0<w<1$)
\begin{align}\label{PDF-W}
	f_W(w)
=
{w^{{1\over 2 r}-1}\over 4 r} 
\left[
f_{X\over X+Y}\left(w_r\right)
+
f_{X\over X+Y}\left(1-w_r\right)
\right].
\end{align}

\bigskip 
Analogously, by taking $a\to-\infty$ and $b=z\in(0,1)$ in \eqref{law-Z}, we have $a_*=0$ and $b_*=z$. Then the  CDF of $Z$ can be written as  
\begin{align}
	F_Z(z)\label{exp-f-z}
	=
	\int_{1-z_r}^{z_r}
	f_{X\over X+Y}(u){\rm d}u,
	\quad 
	0<z<1,
\end{align}
where we are using the notation
\begin{align}\label{def-z-r}
	z_r\equiv {1+\sqrt{1-(1-z)^{1\over r}}\over 2}.
\end{align}
Hence, the PDF of $Z$ is given by (for 	$0<z<1$)
\begin{align}\label{PDF-Z}
	f_Z(z)
	=
	{(1 - z)^{{1\over r}-1}\over 4r\sqrt{1 - (1 - z)^{1\over r}}}
	\left[
	f_{X\over X+Y}\left(z_r\right)
	+
	f_{X\over X+Y}\left(1-z_r\right)
	\right].
\end{align}


Plugging \eqref{formula-pdf-gammas} into formulas \eqref{PDF-W} and \eqref{PDF-Z}, we get (for $0<w,z<1$)
\begin{align}\label{PDF-W-gammas}
	f_W(w)
	=
	{1\over 4 r \lambda_1^{\alpha_2}\lambda_2^{\alpha_1} \text{B}(\alpha_1,\alpha_2)} 
	\,
	w^{{1\over 2 r}-1} G(w_r)
\end{align}
and
\begin{align}\label{PDF-Z-gammas}
	f_Z(z)
	&=
	{1\over 4r\lambda_1^{\alpha_2}\lambda_2^{\alpha_1} \text{B}(\alpha_1,\alpha_2)} \,
	{(1 - z)^{{1\over r}-1} G(z_r)\over \sqrt{1 - (1 - z)^{1\over r}}},
\end{align}
respectively, where
\begin{align}\label{def-G}
G(x)
\equiv
	\dfrac{\displaystyle
	x^{\alpha_1-1}\left(1-x\right)^{\alpha_2-1}}{\displaystyle
	\left[\left({1\over\lambda_2}-{1\over\lambda_1}\right)
	x+{1\over\lambda_1}\right]^{\alpha_1+\alpha_2}}
+
\dfrac{\displaystyle
	\left(1-x\right)^{\alpha_1-1} x^{\alpha_2-1}}{\displaystyle
	\left[\left({1\over\lambda_2}-{1\over\lambda_1}\right)
	(1-x)+{1\over\lambda_1}\right]^{\alpha_1+\alpha_2}}.
\end{align}
Figures \ref{fig:pdf5}-\ref{fig:pdf8} display a range of shapes for the densities of $W$ and $Z$.

%
%
%

\begin{figure}[H]
	\centering
	\includegraphics[width=15.0cm]{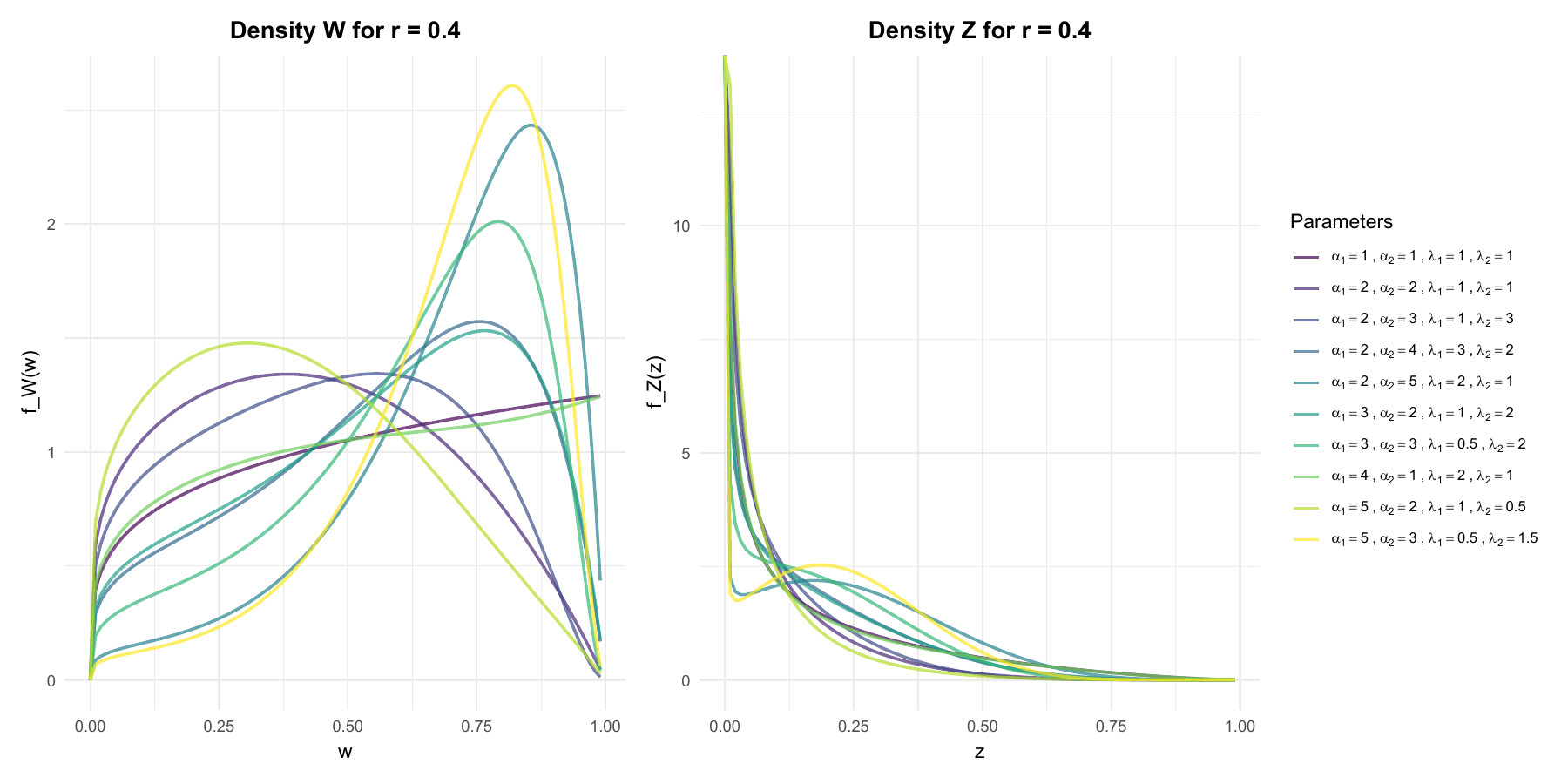}
	\caption{Plots of the PDFs $f_W$ (left) and $f_Z$ (right) with varying parameters $\alpha_1, \alpha_2, \lambda_1$  and $\lambda_2$, with $r=0.4$.}
	\label{fig:pdf5}
\end{figure}


\begin{figure}[H]
	\centering
	\includegraphics[width=15.0cm]{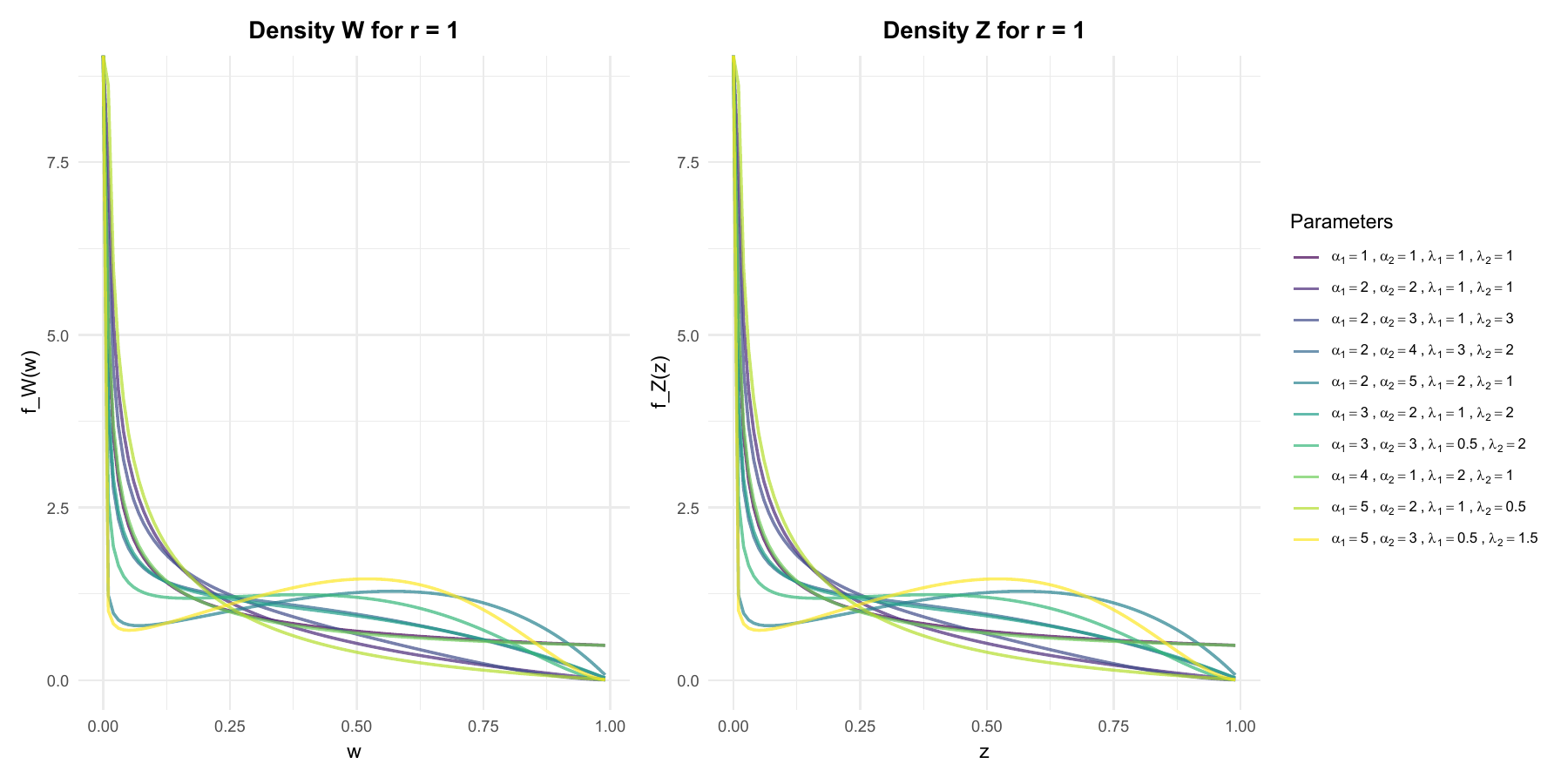}
	\caption{Plots of the PDFs $f_W$ (left) and $f_Z$ (right) with varying parameters $\alpha_1, \alpha_2, \lambda_1$  and $\lambda_2$, with $r=1$.}
	\label{fig:pdf7}
\end{figure}

\begin{figure}[H]
	\centering
	\includegraphics[width=15.0cm]{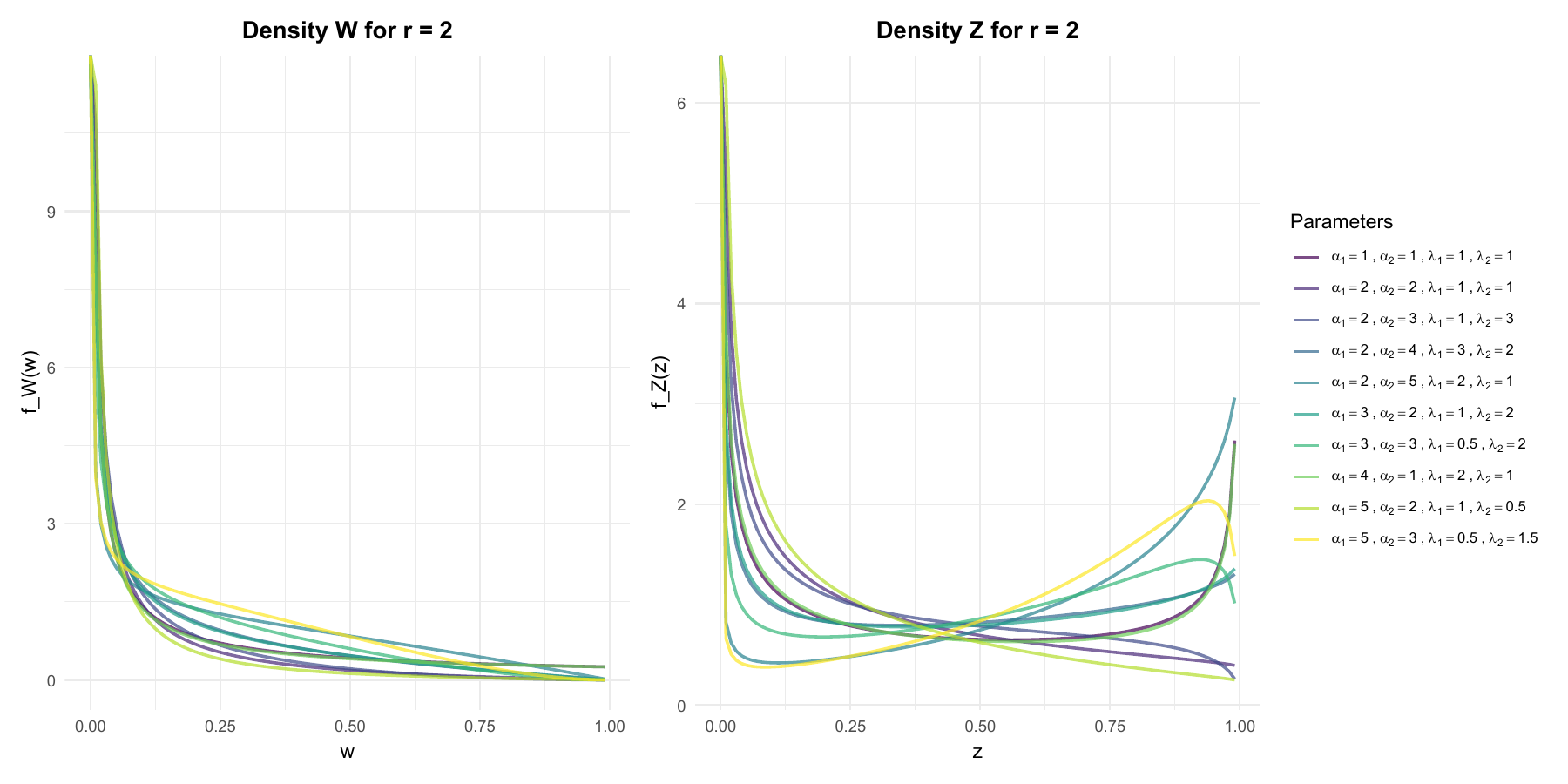}
	\caption{Plots of the PDFs $f_W$ (left) and $f_Z$ (right) with varying parameters $\alpha_1, \alpha_2, \lambda_1$  and $\lambda_2$, with $r=2$.}
	\label{fig:pdf8}
\end{figure}
%
%
%
%

\begin{remark}\label{rem-pdf-equal-lambda}
	In the special case $\lambda_1=\lambda_2$, we get
	\begin{align*}
		f_W(w)
		=
		{w^{{1\over 2 r}-1}\over 4 r  \text{B}(\alpha_1,\alpha_2)} 
		\left[
		{
			\displaystyle
		w_r^{\alpha_1-1}
		(1-w_r)^{\alpha_2-1}
		}
		\right.  
		+
		\left.
		{
			\displaystyle
		(1-w_r)^{\alpha_1-1}
		w_r^{\alpha_2-1}
		}
		\right],
		\quad
		0<w<1,
	\end{align*}
	and
	\begin{align*}
		f_Z(z)
		&=
		{(1 - z)^{{1\over r}-1}\over 4r \text{B}(\alpha_1,\alpha_2)\sqrt{1 - (1 - z)^{1\over r}}}
		\left[
		{\displaystyle
		z_r^{\alpha_1-1}(1-z_r)^{\alpha_2-1}}
		+
		{\displaystyle
		(1-z_r)^{\alpha_1-1} z_r^{\alpha_2-1}}
		\right],
		\quad
		0<z<1,
	\end{align*}
	where $w_r$ and $z_r$ are as given in \eqref{def-w-r} and \eqref{def-z-r}, respectively.
	\end{remark}
	\begin{remark}
	It is worth noting that in the case  $\lambda_1=\lambda_2\equiv\lambda$,
	the densities $f_W$ and $f_Z$ presented in Remark \ref{rem-pdf-equal-lambda} do not depend on $\lambda$, consistent with the fact that $W$ and $Z$, given in \eqref{main-rv-1}, respectively, are derived from scale-invariant transformations. That is, for $X\sim\text{Gamma}(\alpha_1,\lambda)$ and $Y\sim\text{Gamma}(\alpha_2,\lambda)$ we have $\lambda X\sim\text{Gamma}(\alpha_1,1)$ and $\lambda Y\sim\text{Gamma}(\alpha_2,1)$, and consequently, the random variables $W$ and $Z$ satisfy the identities
\begin{align*}
	&W
	=
	\left({Y-X\over X+Y}\right)^{2r}
	=
	\left({\lambda Y-\lambda X\over \lambda X+\lambda Y}\right)^{2r},
	\quad 
	Z
	=
	1-
	{(XY)^r\over \displaystyle \left[{1\over 2}(X+Y)\right]^{2r}}
=
	1-
	{(\lambda X \lambda Y)^r\over \displaystyle \left[{1\over 2}(\lambda X+\lambda Y)\right]^{2r}}, 
\end{align*}
	whose distribution is unaffected by $\lambda$.
\end{remark}

\section{Quantiles}\label{sec:05}

From \eqref{exp-f-w}, the $p$-quantile function $Q_W(p)$ of $W$, for given $p\in(0, 1)$, satisfies the relation
\begin{align*}
	F_W(Q_W(p))
=
F_{X\over X+Y}\left({1+Q_W^{\!^{1\over 2r}}(p)\over 2}\right)
-
F_{X\over X+Y}\left({1-Q_W^{\!^{1\over 2r}}(p)\over 2}\right).
\end{align*}

Analogously, from \eqref{exp-f-z}, the $p$-quantile function $Q_Z(p)$ of $Z$, for given $p\in(0, 1)$, satisfies
\begin{align*}
	F_Z(Q_Z(p))
	=
	F_{X\over X+Y}\left({1+\sqrt{1-[1-Q_Z(p)]^{1\over r}}\over 2}\right)
	-
	F_{X\over X+Y}\left({1-\sqrt{1-[1-Q_Z(p)]^{1\over r}}\over 2}\right).
\end{align*}

While $Q_W(p)$ and $Q_Z(p)$ can be theoretically derived from the above equations, their expressions lack simple closed forms, necessitating numerical methods to invert these equations for quantiles of $W$ and $Z$.

\section{Closed-form expressions for  the $n$th raw moments of $W$ and $Z$}\label{sec:06}

It is clear that, by \eqref{PDF-W-gammas} and \eqref{PDF-Z-gammas}, the $n$th raw moments of $W$ and $Z$ can be calculated by the following formulas:
\begin{align*}
	\mathbb{E}(W^n)
	=
	{1\over 4 r \lambda_1^{\alpha_2}\lambda_2^{\alpha_1} \text{B}(\alpha_1,\alpha_2)}
	\int_0^1 
	w^{n+{1\over 2 r}-1}
G(w_r)
	{\rm d}w,
\quad
	\mathbb{E}(Z^n)
	=
		{1\over 4 r \lambda_1^{\alpha_2}\lambda_2^{\alpha_1} \text{B}(\alpha_1,\alpha_2)}
	\int_0^1  
		{z^n(1 - z)^{{1\over r}-1} G(z_r)\over \sqrt{1 - (1 - z)^{1\over r}}}
	{\rm d}z,
\end{align*}
where $w_r$ and $z_r$ are as given in \eqref{def-w-r} and \eqref{def-z-r}, respectively, and $G(\cdot)$ is as in \eqref{def-G}. Note that the above integrals are well-defined due to the support of $W$ and $Z$ in $(0,1)$, but numerical methods are necessary for their evaluation due to analytical complexity.

The following Subsections \ref{expW}, \ref{rw-m2}, \ref{expW-l} and \ref{rwm-3} provide some analytical expressions for the raw moments of $W$ and $Z$ subject to parameter constraints.

\subsection{The $1$st raw moment of $W$: case $r=1/2$ and $1/3<\lambda_2/\lambda_1<2$
	 }\label{expW}

Under the conditions $r=1/2$ and $1/3<\lambda_2/\lambda_1<2$, in this part, we obtain (more precisely, in \eqref{exp-W-expp-1} and Remark \ref{remark-closed-ep-exp}) straightforward closed-form expressions for 
the $1$st raw moment (expected value) of 
\begin{align*}
	W=\frac{\vert X-Y\vert}{X+Y},
\end{align*}
where $X\sim\text{Gamma}(\alpha_1,\lambda_1)$ and $Y\sim\text{Gamma}(\alpha_2,\lambda_2)$ are independent.

\smallskip
Indeed, as the function $(0,\infty)\ni x\mapsto z\exp(-z x)$, $z>0$, is the density of an exponential distribution, we have
\begin{align}\label{id-pdf-exp}
	\int_{0}^{\infty}\exp(-z x){\rm d}x
	=
	{1\over z}.
\end{align}
By making $z=X+Y$ and by using the fact that $X$ and $Y$ are independent, we get
\begin{align}
	\mathbb{E}(W)
	=
	\mathbb{E}\left[\frac{\vert X-Y\vert}{X+Y}\right]
	=
	\mathbb{E}\left[\vert X-Y\vert \int_{0}^{\infty}\exp\left\{-\left(X+Y\right) x\right\}{\rm d}x\right]
=
	\int_{0}^{\infty}
	\mathbb{E}\left[
	\vert X-Y\vert 
	\exp\left\{-(X+Y) x\right\} 
	\right]{\rm d}x,
	\label{eq-1} 
\end{align}
where in the last step, we applied the Tonelli's theorem \citep{Folland1999} to justify the change in the order of integration.

On the other hand, note that
\begin{align}
		\mathbb{E}\left[
	\vert X-Y\vert 
	\exp\left\{-(X+Y) x\right\} 
	\right]
	&=
	\int_{0}^{\infty}
	\left[
	\int_{0}^{\infty}
	(v-u)\exp(-vx)
	{\rm d}F_Y(v)
	\right]
	\exp(-ux)
	{\rm d}F_X(u)
		\nonumber
	\\[0,2cm]
	&-
	\int_{0}^{\infty}
	\left[
	\int_{0}^{u}
	(v-u)\exp(-vx)
	{\rm d}F_Y(v)
	\right]
	\exp(-ux)
	{\rm d}F_X(u)
	\nonumber
	\\[0,2cm]
&+
\int_{0}^{\infty}
\left[
\int_{0}^{u}
(u-v)\exp(-vx)
{\rm d}F_Y(v)
\right]
\exp(-ux)
{\rm d}F_X(u).
\label{id-first}
\end{align}
Since
\begin{align*}
\int_{0}^{\infty}
\Bigg[
\int_{0}^{\infty}
(v-u)\exp(-vx)
{\rm d}F_Y(v)
\Bigg]
\exp(-ux)
=
\mathbb{E}\left[\exp(-Xx)\right]\mathbb{E}\left[Y\exp(-Yx)\right]
-
\mathbb{E}\left[X\exp(-Xx)\right]\mathbb{E}\left[\exp(-Yx)\right]
\end{align*}
and
\begin{align*}
		&-
	\int_{0}^{\infty}
	\left[
	\int_{0}^{u}
	(v-u)\exp(-vx)
	{\rm d}F_Y(v)
	\right]
	\exp(-ux)
	{\rm d}F_X(u)
+
	\int_{0}^{\infty}
	\left[
	\int_{0}^{u}
	(u-v)\exp(-vx)
	{\rm d}F_Y(v)
	\right]
	\exp(-ux)
	{\rm d}F_X(u)
		\\[0,2cm]
	&=
	2
	\int_{0}^{\infty}
\left[
\int_{0}^{u}
\exp(-vx)
{\rm d}F_Y(v)
\right]
u
\exp(-ux)
{\rm d}F_X(u)
-
	2
\int_{0}^{\infty}
\left[
\int_{0}^{u}
v\exp(-vx)
{\rm d}F_Y(v)
\right]
\exp(-ux)
{\rm d}F_X(u),
\end{align*}
the identity in \eqref{id-first} becomes
\begin{align}
&		\mathbb{E}\left[
\vert X-Y\vert 
\exp\left\{-(X+Y) x\right\} 
\right]
=
\mathbb{E}\left[\exp(-Xx)\right]\mathbb{E}\left[Y\exp(-Yx)\right]
-
\mathbb{E}\left[X\exp(-Xx)\right]\mathbb{E}\left[\exp(-Yx)\right]
\nonumber
\\[0,2cm]
&+
	2
\int_{0}^{\infty}
\left[
\int_{0}^{u}
\exp(-vx)
{\rm d}F_Y(v)
\right]
u
\exp(-ux)
{\rm d}F_X(u)
-
2
\int_{0}^{\infty}
\left[
\int_{0}^{u}
v\exp(-vx)
{\rm d}F_Y(v)
\right]
\exp(-ux)
{\rm d}F_X(u),
\label{key-id}
\end{align}
where
$X\sim\text{Gamma}(\alpha_1,\lambda_1)$ and $Y\sim\text{Gamma}(\alpha_2,\lambda_2)$.

By using results of Propositions \ref{prop-a-1} and \ref{prop-a-2} in \eqref{key-id}, it follows that
\begin{align}
	&\mathbb{E}\left[
	\vert X-Y\vert 
	\exp\left\{-(X+Y) x\right\} 
	\right]
	=
	{\alpha_2\lambda_1^{\alpha_1}\lambda_2^{\alpha_2}\over(x+\lambda_1)^{\alpha_1}(x+\lambda_2)^{\alpha_2+1}}
	-
	{\alpha_1\lambda_1^{\alpha_1}\lambda_2^{\alpha_2}\over (x+\lambda_1)^{\alpha_1+1}(x+\lambda_2)^{\alpha_2}} 
	\nonumber
	\\[0,2cm]
	&+
	{2\lambda_1^{\alpha_1}\lambda_2^{\alpha_2}\over \text{B}(\alpha_1,\alpha_2+1)} \,
	{\,_{2}F_{1}\left(\alpha_1+\alpha_2+1,1;\alpha_2+1;{x+\lambda_2\over 2x+\lambda_1+\lambda_2}\right)\over (2x+\lambda_1+\lambda_2)^{\alpha_1+\alpha_2+1}}
	\nonumber
	\\[0,2cm]
	&-
	{2\lambda_1^{\alpha_1}\lambda_2^{\alpha_2}\over \text{B}(\alpha_1,\alpha_2)} \,
	{\,_{2}F_{1}\left(\alpha_1+\alpha_2,1;\alpha_2+1;{x+\lambda_2\over 2x+\lambda_1+\lambda_2}\right)\over(x+\lambda_2)(2x+\lambda_1+\lambda_2)^{\alpha_1+\alpha_2}}
+
	{2\lambda_1^{\alpha_1}\lambda_2^{\alpha_2}\over\text{B}(\alpha_1,\alpha_2)
	} \,
	{1\over (x+\lambda_2)(2x+\lambda_1+\lambda_2)^{\alpha_1+\alpha_2}}.
	\label{iidd-int}
\end{align}
By replacing \eqref{iidd-int} in \eqref{eq-1}, we get
\begin{align}
\mathbb{E}(W)
	&=
	\alpha_2\lambda_1^{\alpha_1}\lambda_2^{\alpha_2}
		\int_{0}^{\infty}
{1\over(x+\lambda_1)^{\alpha_1}(x+\lambda_2)^{\alpha_2+1}}
{\rm d}x
-
\alpha_1\lambda_1^{\alpha_1}\lambda_2^{\alpha_2}
	\int_{0}^{\infty}
{1\over (x+\lambda_1)^{\alpha_1+1}(x+\lambda_2)^{\alpha_2}} 
{\rm d}x
\nonumber
\\[0,2cm]
&+
{2\lambda_1^{\alpha_1}\lambda_2^{\alpha_2}\over \text{B}(\alpha_1,\alpha_2+1)} 
	\int_{0}^{\infty}
{\,_{2}F_{1}\left(\alpha_1+\alpha_2+1,1;\alpha_2+1;{x+\lambda_2\over 2x+\lambda_1+\lambda_2}\right)\over (2x+\lambda_1+\lambda_2)^{\alpha_1+\alpha_2+1}}
{\rm d}x
\nonumber
\\[0,2cm]
&-
{2\lambda_1^{\alpha_1}\lambda_2^{\alpha_2}\over \text{B}(\alpha_1,\alpha_2)} 
	\int_{0}^{\infty}
{\,_{2}F_{1}\left(\alpha_1+\alpha_2,1;\alpha_2+1;{x+\lambda_2\over 2x+\lambda_1+\lambda_2}\right)\over(x+\lambda_2)(2x+\lambda_1+\lambda_2)^{\alpha_1+\alpha_2}}
{\rm d}x
+
{2\lambda_1^{\alpha_1}\lambda_2^{\alpha_2}\over\text{B}(\alpha_1,\alpha_2)
} 
	\int_{0}^{\infty}
{1\over (x+\lambda_2)(2x+\lambda_1+\lambda_2)^{\alpha_1+\alpha_2}}
{\rm d}x.
\label{exp-W-expp}
\end{align}

By using Proposition \ref{main-prop}, the expectation of $W$ in \eqref{exp-W-expp} can be written as
{\small
\begin{align}
	\mathbb{E}(W)
	&=
{\alpha_2\lambda_2^{\alpha_2}\over \lambda_1^{\alpha_2}(\alpha_1+\alpha_2)}
\, _2F_1\left(\alpha_2+1,\alpha_1+\alpha_2;\alpha_1+\alpha_2+1;1-{\lambda_2\over \lambda_1}\right)
	-
{\alpha_1\lambda_2^{\alpha_2}\over \lambda_1^{\alpha_2}(\alpha_1+\alpha_2)}
\, _2F_1\left(\alpha_2,\alpha_1+\alpha_2;\alpha_1+\alpha_2+1;1-{\lambda_2\over \lambda_1}\right)
	\nonumber
	\\[0,2cm]
	&+
	{2\lambda_1^{\alpha_1}\lambda_2^{\alpha_2}\over \text{B}(\alpha_1,\alpha_2+1)} 
	\int_{0}^{\infty}
	{\,_{2}F_{1}\left(\alpha_1+\alpha_2+1,1;\alpha_2+1;{x+\lambda_2\over 2x+\lambda_1+\lambda_2}\right)\over (2x+\lambda_1+\lambda_2)^{\alpha_1+\alpha_2+1}}
	{\rm d}x
	-
	{2\lambda_1^{\alpha_1}\lambda_2^{\alpha_2}\over \text{B}(\alpha_1,\alpha_2)} 
	\int_{0}^{\infty}
	{\,_{2}F_{1}\left(\alpha_1+\alpha_2,1;\alpha_2+1;{x+\lambda_2\over 2x+\lambda_1+\lambda_2}\right)\over(x+\lambda_2)(2x+\lambda_1+\lambda_2)^{\alpha_1+\alpha_2}}
	{\rm d}x
	\nonumber
	\\[0,2cm]
	&+
{\lambda_1^{\alpha_1}\over 2^{\alpha_1+\alpha_2-1}\lambda_2^{\alpha_1} (\alpha_1+\alpha_2) \text{B}(\alpha_1,\alpha_2)}
\, _2F_1\left(\alpha_1+\alpha_2,\alpha_1+\alpha_2;\alpha_1+\alpha_2+1;1-{\lambda_1+\lambda_2\over 2\lambda_2}\right),
	\label{exp-W-expp-1}
\end{align}
}\noindent
provided $1/3<\lambda_2/\lambda_1<2$.

\begin{remark}
It is noteworthy that, for $\lambda_1\neq\lambda_2$, the integrals in \eqref{exp-W-expp-1}
lack closed-form representations in terms of standard special functions, but can be efficiently computed numerically. For the special case $\lambda_1=\lambda_2$, Remark \ref{remark-closed-ep-exp} provides a straightforward closed-form expression for the expected value of $W$.
\end{remark}

\begin{remark}\label{remark-closed-ep-exp}	
	Setting  $\lambda_1=\lambda_2$ in \eqref{exp-W-expp-1} and using Proposition \ref{main-prop}, we obtain.
	\begin{multline*}	
		\mathbb{E}(W)
		=
	{\alpha_2-\alpha_1\over\alpha_1+\alpha_2}
	+
	{1\over 2^{\alpha_1+\alpha_2-1} (\alpha_1+\alpha_2)}
		\\[0,2cm]
		\times 
			{1\over \text{B}(\alpha_1,\alpha_2)}
\left[		
1
	+
	{\alpha_1+\alpha_2\over 2 \alpha_2} 
	\,_{2}F_{1}\left(\alpha_1+\alpha_2+1,1;\alpha_2+1;{1\over 2}\right)
	\right. 
	-
	\left.
	\,_{2}F_{1}\left(\alpha_1+\alpha_2,1;\alpha_2+1;{1\over 2}\right)
	\right],
	\end{multline*}	
	because $\, _2F_1(a,b;c;0)=1$ and $\text{B}(\alpha_1,\alpha_2+1)=\alpha_2 \text{B}(\alpha_1,\alpha_2)/(\alpha_1+\alpha_2)$. If furthermore, $\alpha_1=\alpha_2\equiv\alpha$, the above expression reduces to
		\begin{align}\label{id-1-0}
		\mathbb{E}(W)
		=
		{\Gamma(2\alpha)\over 2^{2\alpha}\alpha \Gamma^2(\alpha)}
		\left[		
		1
		+
		\,_{2}F_{1}\left(2\alpha+1,1;\alpha+1;{1\over 2}\right)
		\right. 
		-
		\left.
		\,_{2}F_{1}\left(2\alpha,1;\alpha+1;{1\over 2}\right)
		\right].
	\end{align}
	
	By applying the identity \citep{WolframResearch2024}:
		\begin{align*}
		\,_2F_1\left(a,b;{a+b\over 2};{1\over 2}\right)
		=
		\sqrt{\pi}\Gamma\left({a+b\over 2}\right)
		\left[
		{1\over \Gamma({a+1\over 2})\Gamma({b\over 2})}
		+
		{1\over \Gamma({a\over 2})\Gamma({b+1\over 2})}
		\right],
	\end{align*}
we have
\begin{align}\label{id-1-1}
	\,_{2}F_{1}\left(2\alpha+1,1;\alpha+1;{1\over 2}\right)
	=
	1
	+
	{\sqrt{\pi}\Gamma\left(\alpha+1\right)\over \Gamma({2\alpha+1\over 2})}.
\end{align}
	
	Furthermore, by using the identity \citep{WolframResearch2024}:
	\begin{align*}
	\,_2F_1\left(a,b;{a+b+1\over 2};{1\over 2}\right)
	=
	{\sqrt{\pi}\Gamma({a+b+1\over 2})\over\Gamma({a+1\over 2})\Gamma({b+1\over 2})},
\end{align*}
	we get
	\begin{align}\label{id-1-2}
	\,_{2}F_{1}\left(2\alpha,1;\alpha+1;{1\over 2}\right)
	=
	{\sqrt{\pi}\Gamma(\alpha+1)\over\Gamma({2\alpha+1\over 2})}.
	\end{align}
	
	Replacing  \eqref{id-1-1} and \eqref{id-1-2} in \eqref{id-1-0}, yields
	\begin{align*}
	\mathbb{E}(W)
	=
	{2\Gamma(2\alpha)\over 2^{2\alpha}\alpha \Gamma^2(\alpha)}
	=
	{\Gamma(2\alpha+1)\over \alpha^2  2^{2\alpha} \Gamma^2(\alpha)},
	\end{align*}
	where in the last step we have used the well-known identity $\Gamma(x+1)=x\Gamma(x)$. By applying Legendre duplication formula  \citep{Abramowitz1972}:
	$
	\Gamma(x)\Gamma\left(x+{1/ 2}\right)=2^{1-2x}\sqrt{\pi}\Gamma(2x),
	$ in the above identity,
	we get 
	\begin{align*}
	\mathbb{E}(W)
		=
	{\Gamma(\alpha+{1\over 2})\over\sqrt{\pi}\alpha\Gamma(\alpha)},
	\end{align*}
which corresponds to the Gini coefficient, $G$, formula 
for gamma distributions given in \cite{McDonald1979}.
\end{remark}

\subsection{The $n$th raw moment of $W$: case $r\in\mathbb{N}$ and $\lambda_1=\lambda_2$
}\label{rw-m2}

By using Proposition \ref{Moments-truncated}, with $\varepsilon=0$ and $\delta=\infty$, $\mathbb{E}(W^n)$ can be written as
\begin{align*}
	\mathbb{E}(W^n)
	&=
	n
	\int_0^1 w^{n-1} [1-F_W(w)]{\rm d}w
=
	n
	\int_0^1 w^{n-1} \left[1-	\int_{1-w_r}^{w_r} f_{X\over X+Y}(u){\rm d}u\right]{\rm d}w	
	\\[0,2cm]
	&=
	4rn
	\int_{1/2}^1
	(2y-1)^{2rn-1}
	\left[1-\int_{1-y}^{y} f_{X\over X+Y}(u){\rm d}u\right]{\rm d}y,
\end{align*}
where in the second equality the formula \eqref{exp-f-w} for $F_W$ was used, and in the third equality the change of variables $y=w_r={(1+w^{\!^{1/2r}})/ 2}$.
Let $r\in\mathbb{N}$. By the Binomial Theorem, the above expression becomes
\begin{align*}
	4rn
	\sum_{k=0}^{2rn-1}
	\binom{2rn-1}{k}
	2^k(-1)^{2rn-1-k}
\xi(k),
\end{align*}
where 
$\xi(p)
\equiv 	
\int_{1/2}^1
y^p
[1-\int_{1-y}^{y} f_{X\over X+Y}(u){\rm d}u]{\rm d}y$, as defined in \eqref{def-xi}.
As $\lambda_1=\lambda_2$, by using Proposition \ref{prop-app-main},  the $n$th raw moment of $W$ can be written as
\begin{multline*}
	\mathbb{E}(W^n)
	=
	4rn
	\sum_{k=0}^{2rn-1}
	\binom{2rn-1}{k}
	{2^k(-1)^{2rn-1-k}\over k+1}
	\left\{
	{\mathrm {B}(\alpha_1+k+1,\alpha_2)\over \mathrm {B}(\alpha_1,\alpha_2)}
	\left[
	1-I_{1/2}(\alpha_1+k+1,\alpha_2)
	\right]
	\right.
	\\[0,2cm]
	+\left.
	{\mathrm {B}(\alpha_1,\alpha_2+k+1)\over \mathrm {B}(\alpha_1,\alpha_2)} \,
	I_{1/2}(\alpha_1,\alpha_2+k+1)
	-
	\left({1\over 2}\right)^{k+1}
	\right\},
\end{multline*}
provided $r\in\mathbb{N}$.  In the above, $I_x(a,b)={\rm B}(x;a,b)/{\rm B}(a,b)$ denotes the regularized beta function, with ${\rm B}(x;a,b)$ being the  incomplete beta function.

\subsection{The $1$st raw moment of $Z$: case $r=1/2$ and $\lambda_2/\lambda_1<2$
}\label{expW-l}

Under the condition $r=1/2$, in this subsection (more specifically, in \eqref{exp-id-Z} and Remark \ref{rem-expect-Z-lambdas-equal}), we obtain analytical formulas for
the $1$st raw moment (expected value) of 
\begin{align*}
	Z=1-{\sqrt{XY}\over {1\over 2}(X+Y)},
\end{align*}
where $X\sim\text{Gamma}(\alpha_1,\lambda_1)$ and $Y\sim\text{Gamma}(\alpha_2,\lambda_2)$ are independent.

\smallskip
Indeed, by using the identity \eqref{id-pdf-exp} with $z=1/[(X+Y)/2]$, we have
\begin{align*}
	\mathbb{E}(Z)
	=
	1-
	\mathbb{E}\left[{\sqrt{XY}\over {1\over 2}(X+Y)}\right]
	=
	1-
	\mathbb{E}\left[
	\sqrt{XY}
	\int_{0}^{\infty}\exp\left\{-\left({X+Y\over 2}\right) x\right\} 
	{\rm d}x
	\right]
=
	1-
		\int_{0}^{\infty}
	\mathbb{E}\left[	\sqrt{XY}
	\exp\left\{-\left({X+Y\over 2}\right)x\right\}\right]
	{\rm d}x,
\end{align*}
where we used Tonelli's theorem \citep{Folland1999} in the last step to validate the change in integration order. By using the independence of $X$ and $Y$ the above expression becomes
\begin{align}\label{id-1}
		\mathbb{E}(Z)
	=
		1-
	\int_{0}^{\infty}
	\mathbb{E}\left[	\sqrt{X}
	\exp\left\{-{x\over 2} \, X\right\}\right]
		\mathbb{E}\left[	\sqrt{Y}
	\exp\left\{-{x\over 2} \,Y\right\}\right]
	{\rm d}x.
\end{align}

Furthermore, as $X\sim\text{Gamma}(\alpha_1,\lambda_1)$, observe that
\begin{align*}
		\mathbb{E}\left[	\sqrt{X}
	\exp\left\{-{x\over 2} \, X\right\}\right]
	=
	{\lambda_1^{\alpha_1}\over\Gamma(\alpha_1)} \, \int_0^\infty 
	u^{\alpha_1+{1\over 2}-1}
	\exp\left\{-\left({x\over 2}+\lambda_1\right) u\right\}
	{\rm d}u
	=
	{\lambda_1^{\alpha_1}\over\Gamma(\alpha_1)} \,
	{\Gamma(\alpha_1+{1\over 2})\over \left({x\over 2}+\lambda_1\right)^{\alpha_1+{1\over 2}}}.
\end{align*}
Analogously,  as $Y\sim\text{Gamma}(\alpha_2,\lambda_2)$,
\begin{align*}
	\mathbb{E}\left[	\sqrt{Y}
	\exp\left\{-{x\over 2} \, Y\right\}\right]
	=
	{\lambda_2^{\alpha_2}\over\Gamma(\alpha_2)} \,
	{\Gamma(\alpha_2+{1\over 2})\over \left({x\over 2}+\lambda_2\right)^{\alpha_2+{1\over 2}}}.
\end{align*}

Replacing the last two relations in \eqref{id-1}, we get
\begin{align}
	\mathbb{E}(Z)
	&=
	1-
	{\lambda_1^{\alpha_1}\lambda_2^{\alpha_2} \Gamma(\alpha_1+{1\over 2})\Gamma(\alpha_2+{1\over 2})
		\over
	\Gamma(\alpha_1)\Gamma(\alpha_2)} \,
	\int_{0}^{\infty}
{1
	\over 
\left({x\over 2}+\lambda_1\right)^{\alpha_1+{1\over 2}}
\left({x\over 2}+\lambda_2\right)^{\alpha_2+{1\over 2}}
}
	{\rm d}x
	\nonumber
	\\[0,2cm]
	&=
		1-
	{2\lambda_2^{\alpha_2} \Gamma(\alpha_1+{1\over 2})\Gamma(\alpha_2+{1\over 2})
		\over
		\lambda_1^{\alpha_2}(\alpha_1+\alpha_2)
		\Gamma(\alpha_1)\Gamma(\alpha_2)}  \,
	\, _2F_1\left(\alpha_2+{1\over 2},\alpha_1+\alpha_2;\alpha_1+\alpha_2+1;1-{\lambda_2\over\lambda_1}\right),
\label{exp-id-Z}
\end{align}
whenever $\lambda_2/\lambda_1<2$,
where in the final equality, Proposition \ref{main-prop} was applied.

\begin{remark}\label{rem-expect-Z-lambdas-equal}
	From \eqref{exp-id-Z}, for $\lambda_1=\lambda_2$, we have
\begin{align*}	
	\mathbb{E}(Z)
=
1-
{2\Gamma(\alpha_1+{1\over 2})\Gamma(\alpha_2+{1\over 2})
	\over
	(\alpha_1+\alpha_2)
	\Gamma(\alpha_1)\Gamma(\alpha_2)},
\end{align*}	
	because $\, _2F_1(a,b;c;0)=1$.
\end{remark}

\subsection{The $n$th raw moment of $Z$: case $r\in\mathbb{N}$ and $\lambda_1=\lambda_2$
}\label{rwm-3}

Applying Proposition \ref{Moments-truncated}, with $\varepsilon = 0$ and $\delta = \infty$, $\mathbb{E}(Z^n)$ becomes
\begin{align*}
	\mathbb{E}(Z^n)
	&=
	n
	\int_0^1 z^{n-1} [1-F_Z(z)]{\rm d}z
=
	n
	\int_0^1 z^{n-1} 
	\left[1-	\int_{1-z_r}^{z_r} f_{X\over X+Y}(u){\rm d}u\right] 
	{\rm d}z
	\\[0,2cm]
	&=
	4rn
	\int_{1/2}^1
	[1-2^r(2y^2-y+1)^r]^{n-1}
	\left[1-\int_{1-y}^{y} 
	f_{X\over X+Y}(u) 
	{\rm d}u\right]{\rm d}y,
\end{align*}
where in the second equality the formula \eqref{def-z-r} for $F_Z$ was used, and in the third equality the change of variables $y=z_r= {[1+\sqrt{1-(1-z)^{1/r}}]/2}$.
As $r\in\mathbb{N}$, by the multinomial Theorem the above expression is written as
\begin{align*}
	4rn
	\sum_{k=0}^{n-1}
	\binom{n-1}{k}
	(-1)^k 2^{rk}
	\sum_{i+j+l=rk}
	\binom{rk}{i,j,l}
	2^i(-1)^j
	\xi(2i+j),
\end{align*}
where $\xi(\cdot)$ is as considered in Subsection \ref{rw-m2} and 
$
{n \choose k_{1},k_{2},\ldots ,k_{m}}={n!}/(k_{1}!\,k_{2}!\cdots k_{m}!)$
is a multinomial coefficient.
As $\lambda_1=\lambda_2$, by using Proposition \ref{prop-app-main},  the $n$th raw moment of $Z$ can be written as
\begin{multline*}
	\mathbb{E}(Z^n)
	=
	4rn
\sum_{k=0}^{n-1}
\binom{n-1}{k}
\sum_{i+j+l=rk}
\binom{rk}{i,j,l}
{2^{i+rk} (-1)^{j+k}\over 2i+j}
\Bigg\{
{\mathrm {B}(\alpha_1+2i+j+1,\alpha_2)\over \mathrm {B}(\alpha_1,\alpha_2)}
\left[
1-I_{1/2}(\alpha_1+2i+j+1,\alpha_2)
\right]
\\[0,2cm]
+
{\mathrm {B}(\alpha_1,\alpha_2+2i+j+1)\over \mathrm {B}(\alpha_1,\alpha_2)} \,
I_{1/2}(\alpha_1,\alpha_2+2i+j+1)
-
\left({1\over 2}\right)^{2i+j+1}
\Bigg\},
\end{multline*}
provided $r\in\mathbb{N}$. 
Here, $I_x(a,b)$ denotes the regularized beta function.

\section{Maximum likelihood estimation}\label{sec:07}

Let $\{W_i:i = 1,\ldots, n\}$ and $\{Z_i:i = 1,\ldots, n\}$ be univariate random samples of size $n$ from $W$ and $Z$ with density functions as given in \eqref{PDF-W-gammas} and \eqref{PDF-Z-gammas}, respectively, and let $w_i$ and $z_i$ be the sample observations of $W_i$ and $Z_i$, respectively. Then, the respective log-likelihood functions for $\boldsymbol{\theta}=(\alpha_1,\alpha_2,\lambda_1,\lambda_2,r)^\top$  are given by
\begin{align}\label{PDF-W-gammas-log}
	\ell(\boldsymbol{\theta};\boldsymbol{w})
	&=
	-
	\log(4)
	-
	\log(r)
	-{\alpha_2}\log(\lambda_1)
	-{\alpha_1}\log(\lambda_2)
	-\log(\text{B}(\alpha_1,\alpha_2))
	+
	\left({{1\over 2 r}-1}\right)\sum_{i=1}^{n}\log(w_i)
	+
	\sum_{i=1}^{n}
	\log(G(w_{r,i})),
\end{align}
and
\begin{align}\label{PDF-Z-gammas-log}
	\ell(\boldsymbol{\theta};\boldsymbol{z})
	&=
	-
	\log(4)
	-
	\log(r)
	-{\alpha_2}\log(\lambda_1)
	-{\alpha_1}\log(\lambda_2)
	-\log(\text{B}(\alpha_1,\alpha_2))
		\nonumber
	\\[0,2cm]
	&
	+
	\left({{1\over r}-1}\right)\sum_{i=1}^{n}\log(1 - z_i)
	-
	{1\over 2}
	\sum_{i=1}^{n}
	\log\left({1 - (1 - z_i)^{1\over r}}\right)
	+
	\sum_{i=1}^{n}
	\log(G(z_{r,i})),
\end{align}
where
\begin{align}\label{def-wri}
	w_{r,i}={1+w_i^{\!^{1\over 2r}}\over 2},
	\quad 
	z_{r,i}={1+\sqrt{1-(1-z_i)^{1\over r}}\over 2},
	\quad 
	i=1,\ldots,n,
\end{align}
and $G(\cdot)$ being as in \eqref{def-G}. The necessary conditions for the occurrence
of a maximum (or a minimum) of $\ell(\boldsymbol{\theta};\boldsymbol{w})$ and $\ell(\boldsymbol{\theta};\boldsymbol{z})$  are
\begin{align}\label{lik-eq-w}
{\partial \ell(\boldsymbol{\theta};\boldsymbol{w})\over\partial\alpha_i}
=0,
\quad 
{\partial \ell(\boldsymbol{\theta};\boldsymbol{w})\over\partial\lambda_i}
=0,
\quad 
{\partial \ell(\boldsymbol{\theta};\boldsymbol{w})\over\partial r}
=0,
	\quad i=1,2,
\end{align}
and 
\begin{align}\label{lik-eq-z}
	{\partial \ell(\boldsymbol{\theta};\boldsymbol{z})\over\partial\alpha_i}
	=0,
	\quad 
	{\partial \ell(\boldsymbol{\theta};\boldsymbol{z})\over\partial\lambda_i}
	=0,
	\quad 
	{\partial \ell(\boldsymbol{\theta};\boldsymbol{z})\over\partial r}
	=0,
	\quad i=1,2,
\end{align}
respectively. The equations \eqref{lik-eq-w} and \eqref{lik-eq-z} are known as the likelihood equations. It is simple to observe that (for $i=1,2$)
\begin{align*}
	&{\partial \ell(\boldsymbol{\theta};\boldsymbol{w})\over\partial\alpha_1}
	=
	-
	\log(\lambda_2)-[\psi(\alpha_1)-\psi(\alpha_1+\alpha_2)]
		+
	\sum_{i=1}^{n}
	{{\partial G(w_{r,i})\over \partial\alpha_1}\over G(w_{r,i})},
	\\[0,2cm] 
		&{\partial \ell(\boldsymbol{\theta};\boldsymbol{w})\over\partial\alpha_2}
	=
	-
	\log(\lambda_1)-[\psi(\alpha_2)-\psi(\alpha_1+\alpha_2)]
	+
	\sum_{i=1}^{n}
	{{\partial G(w_{r,i})\over \partial\alpha_2}\over G(w_{r,i})},
	\\[0,2cm] 
	&{\partial \ell(\boldsymbol{\theta};\boldsymbol{w})\over\partial\lambda_1}
	=
	-
	{\alpha_2\over\lambda_1}+	\sum_{i=1}^{n}
	{{\partial G(w_{r,i})\over \partial\lambda_1}\over G(w_{r,i})},
	\\[0,2cm] 
		&{\partial \ell(\boldsymbol{\theta};\boldsymbol{w})\over\partial\lambda_2}
	=
		-
	{\alpha_1\over\lambda_2}+	\sum_{i=1}^{n}
	{{\partial G(w_{r,i})\over \partial\lambda_2}\over G(w_{r,i})},
	\\[0,2cm] 
	&{\partial \ell(\boldsymbol{\theta};\boldsymbol{w})\over\partial r}
	=
	-
	{1\over r}
	-
	{1\over 2r^2}
	\sum_{i=1}^{n}
	\log(w_i)	
	+	
	\sum_{i=1}^{n}
	{G'(w_{r,i})\over G(w_{r,i})} \, {\partial w_{r,i}\over \partial r},
\end{align*}
and
\begin{align*}
	&{\partial \ell(\boldsymbol{\theta};\boldsymbol{z})\over\partial\alpha_1}
	=
	-
	\log(\lambda_2)-[\psi(\alpha_1)-\psi(\alpha_1+\alpha_2)]
	+
	\sum_{i=1}^{n}
	{{\partial G(z_{r,i})\over \partial\alpha_1}\over G(z_{r,i})},
	\\[0,2cm] 
	&{\partial \ell(\boldsymbol{\theta};\boldsymbol{z})\over\partial\alpha_2}
	=
	-
	\log(\lambda_1)-[\psi(\alpha_2)-\psi(\alpha_1+\alpha_2)]
	+
	\sum_{i=1}^{n}
	{{\partial G(z_{r,i})\over \partial\alpha_2}\over G(z_{r,i})},
	\\[0,2cm] 
	&{\partial \ell(\boldsymbol{\theta};\boldsymbol{z})\over\partial\lambda_1}
	=
	-
	{\alpha_2\over\lambda_1}+	\sum_{i=1}^{n}
	{{\partial G(z_{r,i})\over \partial\lambda_1}\over G(z_{r,i})},
	\\[0,2cm] 
	&{\partial \ell(\boldsymbol{\theta};\boldsymbol{z})\over\partial\lambda_2}
	=
	-
	{\alpha_1\over\lambda_2}+	\sum_{i=1}^{n}
	{{\partial G(z_{r,i})\over \partial\lambda_2}\over G(z_{r,i})},
	\\[0,2cm] 
	&{\partial \ell(\boldsymbol{\theta};\boldsymbol{z})\over\partial r}
	=
	-
	{1\over r}
	-
	{1\over r^2}
	\sum_{i=1}^{n}
	\log(1-z_i)	
	-
	{1\over 2r^2}
		\sum_{i=1}^{n}
		{(1-z_i)^{1\over r}\log(1-z_i)\over 1-(1-z_i)^{1\over r}}
	+	
	\sum_{i=1}^{n}
	{G'(z_{r,i})\over G(z_{r,i})} \, {\partial z_{r,i}\over \partial r},
\end{align*}
where $G'(x)$ denotes the derivative of $G(x)$, given in  \eqref{def-G}, with respect to $x$, and ${\partial w_{r,i}/ \partial r}$ and ${\partial z_{r,i}/ \partial r}$ are the partial partial derivatives of $w_{r,i}$ and $z_{r,i}$ given in \eqref{def-wri}. In the above, $\psi(x)=\Gamma'(x)/\Gamma(x)$ denotes the digamma function.

Any nontrivial root $\widehat{\boldsymbol{\theta}}=(\widehat{\alpha}_1,\widehat{\alpha}_2,\widehat{\lambda}_1,\widehat{\lambda}_2,\widehat{r})^\top$ of the above likelihood equations that provides the absolute maximum of the log-likelihood function is called an maximum likelihood (ML) estimate. Due to the complexity of the maximization problem, numerical optimization methods are required to obtain the ML estimate.

\section{Simulation Results} \label{simulation}

\subsection{Monte Carlo Simulation}

The assessment of the estimators via Monte Carlo simulation is based on classical performance measures, especially bias, variance, and mean squared error. For an estimator $\widehat{\theta}$ of a parameter $\theta$, as presented in \citep[p.~407--408]{gentle2009computational}, the bias is defined as
$
\mathrm{bias}(\widehat{\theta})
= \mathbb{E}[\widehat{\theta}] - \theta,
$
and the variance as
$
\mathrm{Var}(\widehat{\theta})
=
\mathbb{E}[(\widehat{\theta}-\mathbb{E}[\widehat{\theta}])^2
].
$

The mean squared error (MSE) is given by
$
\mathrm{MSE}(\widehat{\theta})
=
\mathbb{E}[
(\widehat{\theta}-\theta)^2
],
$
and satisfies the classical decomposition in \citep[p.~268--269]{casella2002statistical}:
$
\mathrm{MSE}(\widehat{\theta})
=
\mathrm{Var}(\widehat{\theta})
+
(\mathrm{bias}(\widehat{\theta}))^2.
$

The root mean squared error (RMSE) is defined as
$
\mathrm{RMSE}(\widehat{\theta})
=
\sqrt{\mathrm{MSE}(\widehat{\theta})}.
$

In the Monte Carlo setting, the expectations involved in the expressions above are approximated by sample averages over $B$ independent replications. Denoting by $\widehat{\theta}^{(b)}$ the estimator obtained from replication $b$, the following empirical versions are used:
\[
\widehat{\mathrm{bias}}
=
\frac{1}{B}\sum_{b=1}^{B} \widehat{\theta}^{(b)} - \theta,
\ \
\widehat{\mathrm{Var}}
=
\frac{1}{B-1}
\sum_{b=1}^{B}
\left(
\widehat{\theta}^{(b)} -
\overline{\widehat{\theta}}
\right)^{2},
\ \
\overline{\widehat{\theta}} =
\frac{1}{B}\sum_{b=1}^{B}\widehat{\theta}^{(b)},
\ \
\widehat{\mathrm{MSE}}
=
\frac{1}{B}
\sum_{b=1}^{B}
(\widehat{\theta}^{(b)} - \theta)^2,
\ \
\widehat{\mathrm{RMSE}} = \sqrt{\widehat{\mathrm{MSE}}}.
\] 

The estimation step is central to evaluate parameter stability, sensitivity to different generating configurations, and the role of the transformation parameter $r$. Since the densities $f_W(w)$ and $f_Z(z)$, given in \eqref{PDF-W-gammas} and \eqref{PDF-Z-gammas}, have a complex structure, implementing maximum likelihood estimation requires numerical methods. In this context, Monte Carlo simulation becomes essential to investigate patterns of stability, sensitivity, and convergence of the estimators.

This section provides a synthesis of the results obtained for models $W$ and $Z$, defined in \eqref{main-rv-1}. The analysis highlights representative scenarios, the influence of the structural parameters, and the behavior of the estimators as the sample size increases, allowing one to identify overall performance trends.

Additional scenarios in which the scale parameters are kept fixed, that is, $\lambda_1=\lambda_2=1$, are reported in the Appendix \ref{ap:simulacao}. This separation makes it possible to isolate the effect of the shape parameters and of $r$ on the estimators' performance.

\paragraph*{General simulation setup}

The simulations were based on independent samples of $W$ and $Z$ generated from \eqref{main-rv-1}, with
$
X \sim \text{Gamma}(\alpha_1,\lambda_1)
$ \text{and} $
Y \sim \text{Gamma}(\alpha_2,\lambda_2),
$
as specified in \eqref{formula-pdf-gammas}. In all experiments, the sample sizes
$
n \in \{50, 100, 250, 500\},
$
and transformation parameters
$
r \in \{0.5, 1, 2, 4\}
$
were considered.

To ensure a representative analysis, four combinations of the shape parameters were selected, corresponding to different configurations of $(\alpha_1, \alpha_2)$:
$
(\alpha_1,\alpha_2) \in \{(0.5,0.5), (0.8,1.0), (1.5,1.3), (2.0,2.0)\}.
$

These choices allow the evaluation of a low-shape symmetric case, a moderately asymmetric case, an intermediate asymmetric case, and a high-shape symmetric case. For both models, the five parameters in the vector
$
\boldsymbol{\theta} = (\alpha_1,\alpha_2,\lambda_1,\lambda_2,r)^\top
$
were estimated simultaneously by maximum likelihood, using the L-BFGS-B algorithm with numerical constraints to avoid invalid values during the optimization process. The log-likelihood function in each case follows directly from the densities in \eqref{PDF-W-gammas} and \eqref{PDF-Z-gammas}, and can be written as in \eqref{PDF-W-gammas-log} and \eqref{PDF-Z-gammas-log}, with $w_{r,i}$ and $z_{r,i}$ defined in \eqref{def-wri} and $G(\cdot)$ as in \eqref{def-G}.

Each scenario was replicated 500 times, enabling the computation of bias, standard deviation, and RMSE for each estimator, based on the numerical solutions of the likelihood equations in \eqref{lik-eq-w} and \eqref{lik-eq-z}.

\subsection{Results for model $W$}

Figures~\ref{fig:W-bias-rmse-0505}--\ref{fig:W-bias-rmse-20-20} display, respectively, the bias and RMSE plots of the estimators under model $W$ for the four scenarios considered: $(0.5,0.5)$, $(0.8,1.0)$, $(1.5,1.3)$, and $(2.0,2.0)$. Overall, model $W$ exhibited globally stable performance. Across all scenarios, the estimators $\widehat{\alpha}_1$ and $\widehat{\alpha}_2$ showed progressive convergence as $n$ increased, with decreasing biases and a consistent reduction in RMSE. Even in more challenging settings, such as $(0.5,0.5)$---which, despite being symmetric, is particularly demanding because $\alpha_1=\alpha_2<1$, concentrating probability mass near the boundaries and increasing estimation variability---larger sample sizes substantially reduced the variability of the estimates. Nevertheless, in asymmetric scenarios such as $(0.8,1.0)$ and $(1.5,1.3)$, RMSEs remained relatively high even at $n=500$.

Estimation of the scale parameters $\lambda_1$ and $\lambda_2$ generally performed more favorably than that of the shape parameters. Biases remained relatively small, and RMSEs decreased systematically as $n$ increased. In the low-shape symmetric scenario $(0.5,0.5)$, recovery of the scale parameters is essentially exact, with RMSEs very close to zero for all values of $r$. In the remaining scenarios, including $(0.8,1.0)$, $(1.5,1.3)$, and $(2.0,2.0)$, RMSEs for the scales stayed at a moderate level but were still below those observed for the shape parameters, reinforcing the relatively higher stability of $\widehat{\lambda}_1$ and $\widehat{\lambda}_2$.

Estimation of the transformation parameter $r$ exhibited higher sensitivity, especially in low-shape and strongly asymmetric scenarios. As $n$ increases, however, a clear stabilization of the estimator is observed, with a gradual decline in RMSE. In the symmetric scenario $(2.0,2.0)$, for instance, the error decreases consistently and reaches small values at $n=500$, highlighting the expected asymptotic gain.

\begin{figure}[H]
	\centering
	
	\begin{subfigure}{0.48\linewidth}
		\centering
		\includegraphics[width=\linewidth]{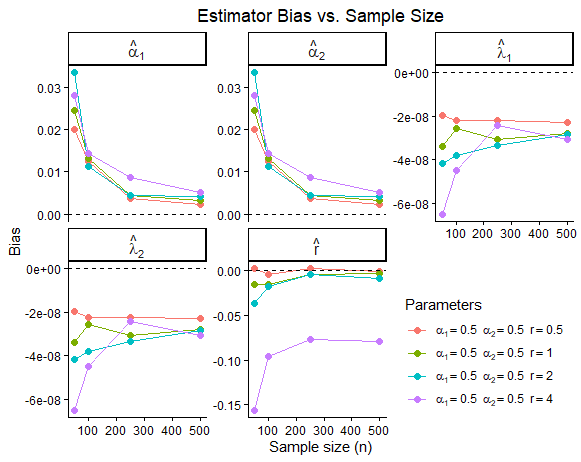}
		\caption{Bias of the estimators}
		\label{fig:W-bias-0505}
	\end{subfigure}
	\hfill
	\begin{subfigure}{0.48\linewidth}
		\centering
		\includegraphics[width=\linewidth]{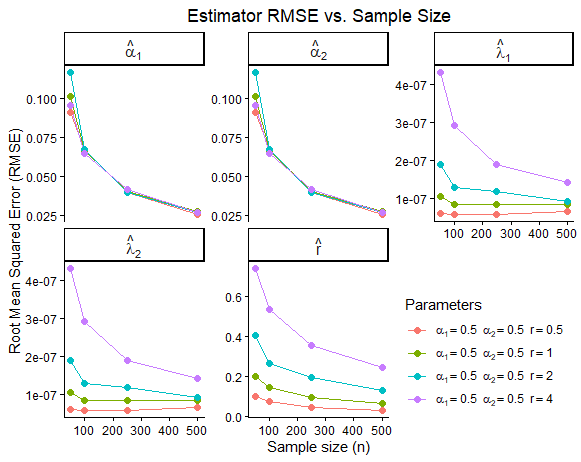}
		\caption{RMSE of the estimators}
		\label{fig:W-rmse-0505}
	\end{subfigure}
	
	\caption{Bias and RMSE of the estimators under model $W$ for $(\alpha_1,\alpha_2)=(0.5,0.5)$.}
	\label{fig:W-bias-rmse-0505}
\end{figure}

\begin{figure}[H]
	\centering
	
	\begin{subfigure}{0.48\linewidth}
		\centering
		\includegraphics[width=\linewidth]{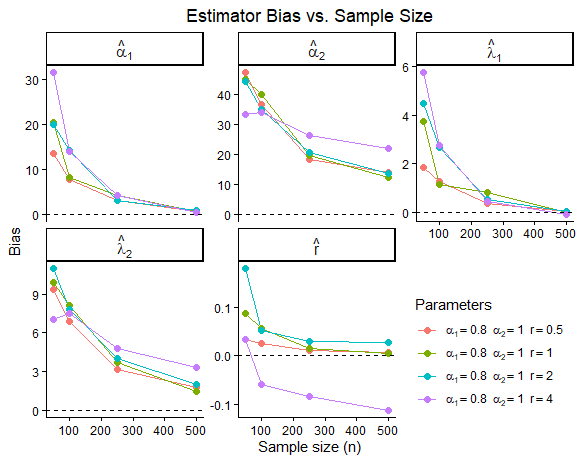}
		\caption{Bias}
		\label{fig:W-rmse-0810}
	\end{subfigure}
	\hfill
	\begin{subfigure}{0.48\linewidth}
		\centering
		\includegraphics[width=\linewidth]{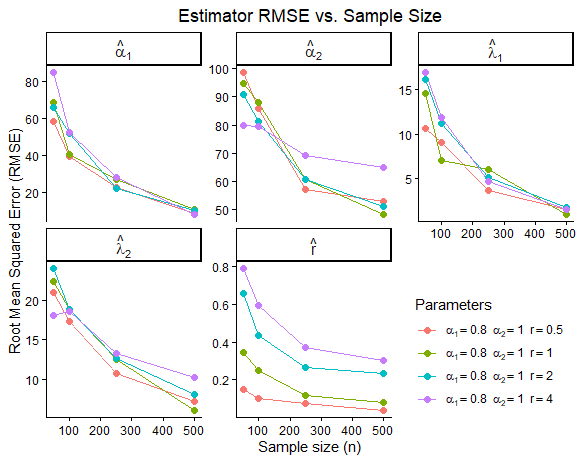}
		\caption{RMSE}
		\label{fig:W-rmse-0810}
	\end{subfigure}
	
	\caption{%
		Bias and RMSE of the estimators under model $W$ for $(\alpha_1,\alpha_2)=(0.8,1.0)$.
	}
	\label{fig:W-bias-rmse-0810}
\end{figure}

\begin{figure}[H]
	\centering
	
	\begin{subfigure}{0.48\linewidth}
		\centering
		\includegraphics[width=\linewidth]{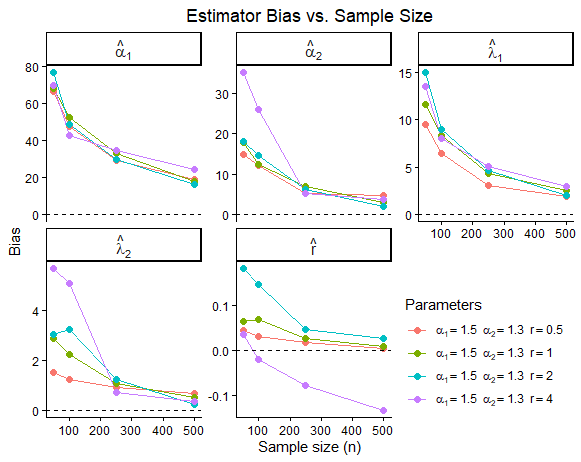}
		\caption{ Bias}
	\end{subfigure}
	\hfill
	\begin{subfigure}{0.48\linewidth}
		\centering
		\includegraphics[width=\linewidth]{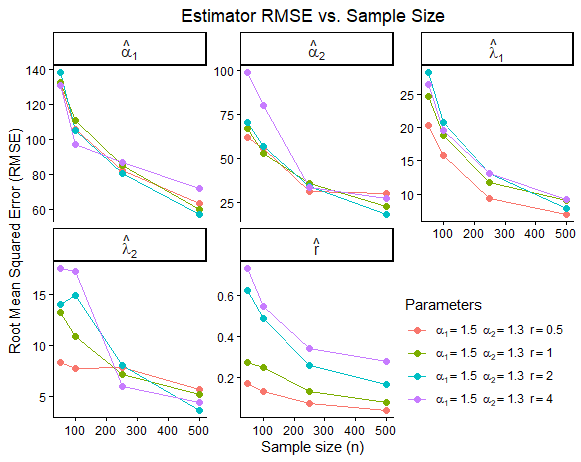}
		\caption{ RMSE}
	\end{subfigure}
	
	\caption{%
		Bias and RMSE of the estimators under model $W$ for $(\alpha_1,\alpha_2)=(1.5,1.3)$.
	}
	\label{fig:W-bias-rmse-1513}
\end{figure}

\begin{figure}[H]
	\centering
	
	\begin{subfigure}{0.48\linewidth}
		\centering
		\includegraphics[width=\linewidth]{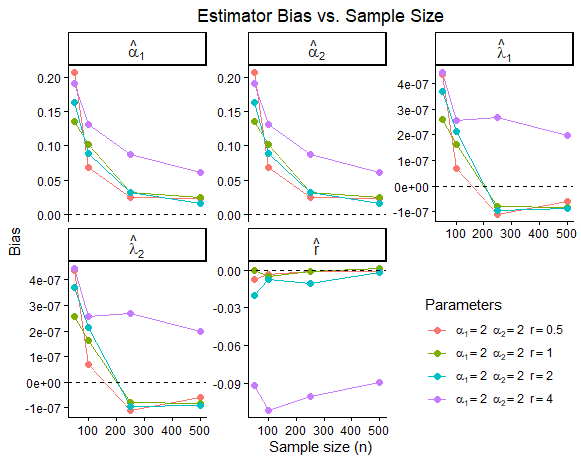}
		\caption{Bias}
	\end{subfigure}
	\hfill
	\begin{subfigure}{0.48\linewidth}
		\centering
		\includegraphics[width=\linewidth]{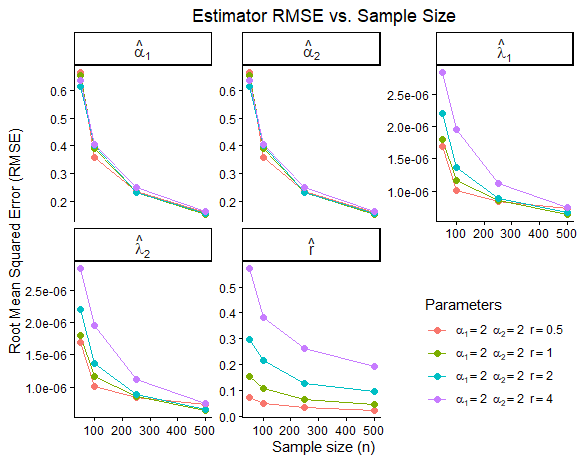}
		\caption{RMSE}
	\end{subfigure}
	
	\caption{%
		Bias and RMSE of the estimators under model $W$ for $(\alpha_1,\alpha_2)=(2.0,2.0)$.
	}
	\label{fig:W-bias-rmse-20-20}
\end{figure}


\subsection{Results for model $Z$}
\label{sec:resultados-Z}

Figures~\ref{fig:Z-bias-rmse-0505}--\ref{fig:Z-bias-rmse-2020} summarize the patterns observed in the four scenarios considered. Model $Z$ proved considerably more challenging than model $W$. The transformation $T_r(x)$ in \eqref{transf-using} and the term
$
\sqrt{1 - (1 - z)^{1/r}}
$
in the density \eqref{PDF-Z-gammas} reduce the concentration of information and increase estimation sensitivity, particularly in asymmetric scenarios or when the shape parameters are below~1. The shape estimators $\widehat{\alpha}_1$ and $\widehat{\alpha}_2$ exhibited large biases and substantial RMSEs across all scenarios. Overestimation of the shape parameters was recurrent both in the low-shape cases $(0.5,0.5)$ and $(0.8,1.0)$ and in the higher-shape scenarios $(1.5,1.3)$ and $(2.0,2.0)$. Although increasing $n$ consistently reduced the errors, RMSEs remained high even at $n=500$, indicating that the structure of model $Z$ imposes limitations on precise identification of the shape parameters.

Estimation of the scale parameters $\lambda_1$ and $\lambda_2$ performed considerably better than estimation of the shape parameters. In the $(0.5,0.5)$ scenario, scale recovery was excellent, with very small RMSEs for all sample sizes. In the remaining scenarios, RMSEs for the scale parameters remained systematically below those for the shape parameters, even though they may increase under unfavorable configurations (for instance, when $r=4$). This behavior suggests that the scale parameters tend to be relatively more stable under model $Z$, albeit only in comparison with the other model components.

Estimation of the transformation parameter $r$ was the most sensitive among all parameters. As $r$ increases, the transformation $T_r(x)$ yields flatter densities, reducing the information available for estimation and markedly increasing RMSE. For moderate values of $r$ ($0.5$ and $1$), performance is superior to that obtained with $r=4$, although precision remains limited. For $r=4$, RMSE remains high even at $n=500$, indicating pronounced instability.

\begin{figure}[H]
	\centering
	\begin{subfigure}{0.48\linewidth}
		\centering
		\includegraphics[width=\linewidth]{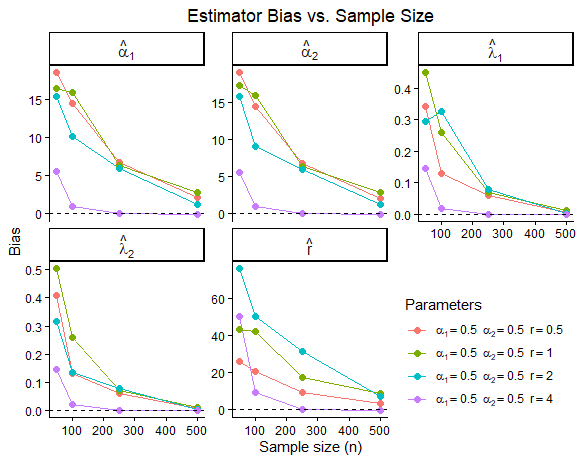}
		\caption{Bias of the estimators}
		\label{fig:Z-bias-0505}
	\end{subfigure}
	\hfill
	\begin{subfigure}{0.48\linewidth}
		\centering
		\includegraphics[width=\linewidth]{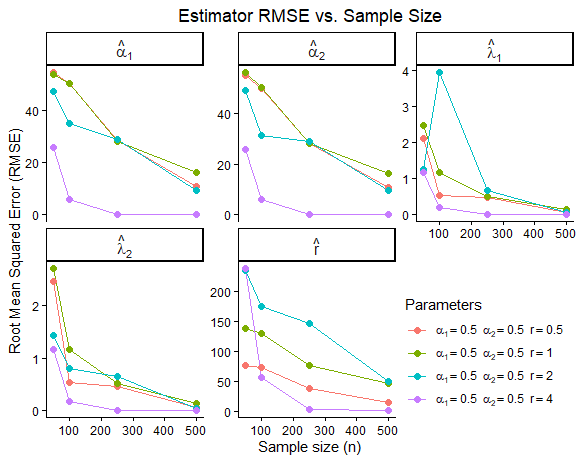}
		\caption{RMSE of the estimators}
		\label{fig:Z-rmse-0505}
	\end{subfigure}
	\caption{Bias and RMSE of the estimators under model $Z$ for $(\alpha_1,\alpha_2)=(0.5,0.5)$, with $n\in\{50,100,250,500\}$ and $r\in\{0.5,1,2,4\}$.}
	\label{fig:Z-bias-rmse-0505}
\end{figure}

\begin{figure}[H]
	\centering
	\begin{subfigure}{0.48\linewidth}
		\centering
		\includegraphics[width=\linewidth]{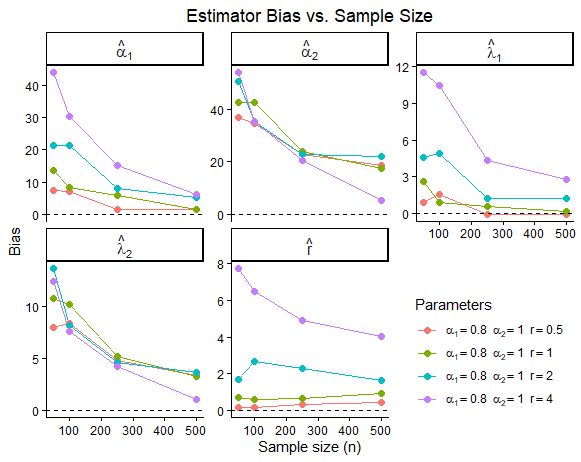}
		\caption{Bias of the estimators}
		\label{fig:Z-bias-0810}
	\end{subfigure}
	\hfill
	\begin{subfigure}{0.48\linewidth}
		\centering
		\includegraphics[width=\linewidth]{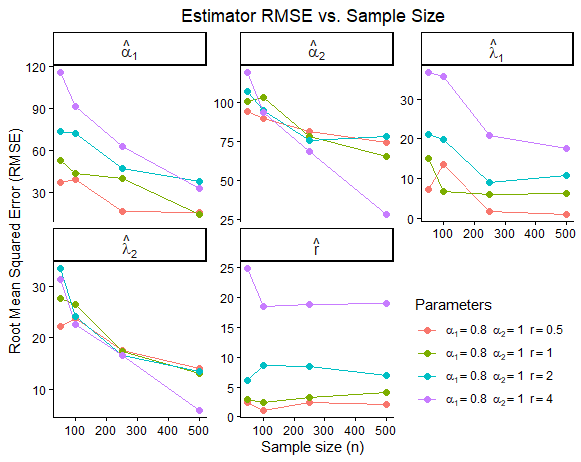}
		\caption{RMSE of the estimators}
		\label{fig:Z-rmse-0810}
	\end{subfigure}
	\caption{Bias and RMSE of the estimators under model $Z$ for $(\alpha_1,\alpha_2)=(0.8,1.0)$, with $n\in\{50,100,250,500\}$ and $r\in\{0.5,1,2,4\}$.}
	\label{fig:Z-bias-rmse-0810}
\end{figure}

\begin{figure}[H]
	\centering
	\begin{subfigure}{0.48\linewidth}
		\centering
		\includegraphics[width=\linewidth]{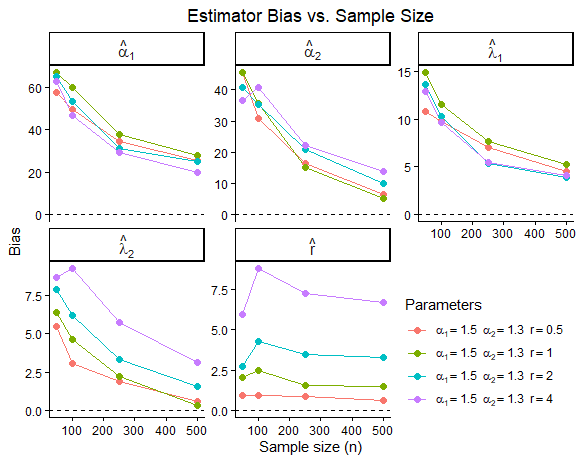}
		\caption{Bias of the estimators}
		\label{fig:Z-bias-1513}
	\end{subfigure}
	\hfill
	\begin{subfigure}{0.48\linewidth}
		\centering
		\includegraphics[width=\linewidth]{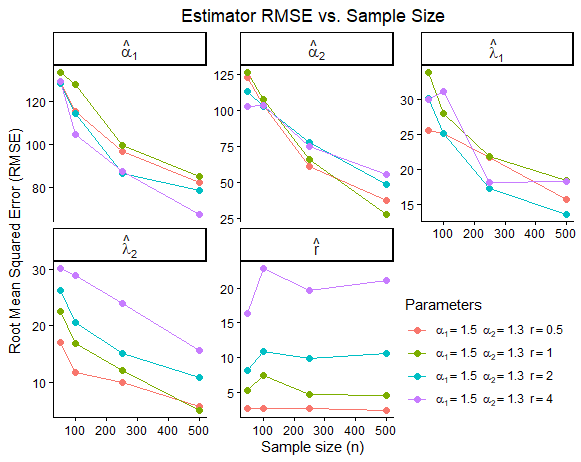}
		\caption{RMSE of the estimators}
		\label{fig:Z-rmse-1513}
	\end{subfigure}
	\caption{Bias and RMSE of the estimators under model $Z$ for $(\alpha_1,\alpha_2)=(1.5,1.3)$, with $n\in\{50,100,250,500\}$ and $r\in\{0.5,1,2,4\}$.}
	\label{fig:Z-bias-rmse-1513}
\end{figure}

\begin{figure}[H]
	\centering
	\begin{subfigure}{0.48\linewidth}
		\centering
		\includegraphics[width=\linewidth]{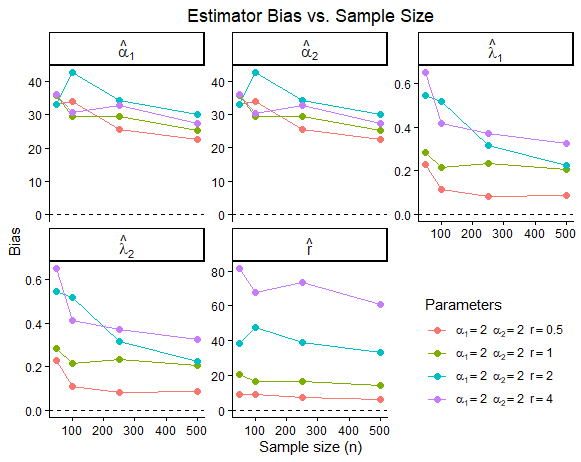}
		\caption{Bias of the estimators}
		\label{fig:Z-bias-2020}
	\end{subfigure}
	\hfill
	\begin{subfigure}{0.48\linewidth}
		\centering
		\includegraphics[width=\linewidth]{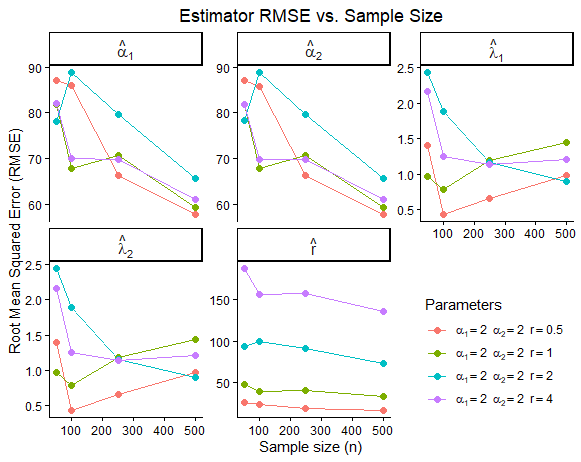}
		\caption{RMSE of the estimators}
		\label{fig:Z-rmse-2020}
	\end{subfigure}
	\caption{Bias and RMSE of the estimators under model $Z$ for $(\alpha_1,\alpha_2)=(2.0,2.0)$, with $n\in\{50,100,250,500\}$ and $r\in\{0.5,1,2,4\}$.}
	\label{fig:Z-bias-rmse-2020}
\end{figure}


\paragraph*{Comparison between the models}
\label{par:comparacao-W-Z}

Models $W$ and $Z$ exhibit distinct behaviors from an estimation standpoint. In general, model $W$ shows greater numerical stability, with both shape and scale estimators displaying smaller biases and RMSEs as the sample size increases. The transformation parameter $r$ is also recovered with good accuracy under model $W$, with substantially smaller errors than those observed under model $Z$, consistently with the structure defined in \eqref{main-rv-1}.

Under model $Z$, sensitivity to asymmetric configurations and very low shape parameters is more pronounced. Recovering $\alpha_1$, $\alpha_2$, and especially $r$ becomes more difficult, indicating limitations in jointly identifying the parameters under low-shape or strongly asymmetric settings. Thus, while model $W$ adapts satisfactorily to a broader range of parametric configurations, model $Z$, defined in \eqref{main-rv-1}, requires larger sample sizes and exhibits more favorable performance in more regular configurations, such as those with symmetry or moderate shape.

\paragraph*{Final remarks on estimation}
\label{par:consideracoes-estimacao}

The results indicate that parameter recovery in the proposed models depends sensitively on the interplay between the generating distribution's shape, the transformation parameter, and the sample size. Under model $W$, the estimators behave robustly and efficiently in most scenarios. For model $Z$, however, maximum likelihood estimation is markedly more sensitive in strongly asymmetric settings, where the precision of the estimates remains limited even for larger samples. These limitations suggest the investigation of alternative estimation methods in future work.


\section{Application to real data} \label{Applications}

This section presents an application of the proposed models using a real data set, with the aim of illustrating the estimation procedure and assessing the goodness of fit of distributions $W$ and $Z$. For this purpose, data on the Gini index for the year 2021 are considered, obtained from the \textit{Our World in Data} repository \citep{owid-economic-inequality}. The Gini index is a continuous variable supported on $(0,1)$, making it well suited for evaluating the performance of the proposed unit-interval models $W$ and $Z$.

After preliminary data processing, a total of 78 countries with available observations for the selected year were identified. Table~\ref{tab:gini_summary_2021} reports the descriptive statistics of the variable under study. The mean ($0.350$) and median ($0.340$) are very close, indicating a fairly symmetric central tendency, although the distribution exhibits moderate right skewness (skewness equal to $0.899$). The coefficient of variation, approximately $18.5\%$, suggests relatively low dispersion, while the positive kurtosis ($0.812$) indicates heavier tails compared to the normal distribution. The observed range, from $0.241$ to $0.551$, confirms that the data lie entirely within the unit interval and display sufficient variability to justify the use of flexible parametric distributions.

\begin{table}[htbp]
	\centering
	\caption{Descriptive statistics for the Gini index data (2021).}
	\label{tab:gini_summary_2021}
	\begin{tabular}{ccccccccc}
		\toprule
		Mean & Median & SD & CV (\%) & Skewness & Kurtosis & Minimum & Maximum & Sample size \\
		\midrule
		0.350 & 0.340 & 0.065 & 18.518\% & 0.899 & 0.812 & 0.241 & 0.551 & 78 \\
		\bottomrule
	\end{tabular}
\end{table}

Figure~\ref{fig:gini-map} displays the thematic map of the spatial distribution of the Gini index for the year 2021. The geographic pattern is striking: countries in Latin America, particularly Brazil, Colombia, and Panama, are concentrated in the upper ranges of the index, reflecting high levels of income inequality. In contrast, most of Western Europe, as well as Canada, South Korea, and Taiwan, exhibit lower Gini values. Across Africa and the Middle East, a wider diversity of inequality levels is observed, with countries spanning several ranges of the indicator. This spatial heterogeneity is consistent with the statistical summary presented in Table~\ref{tab:gini_summary_2021} and further supports the use of models capable of capturing asymmetry and concentration patterns.
\begin{figure}[H]
	\centering
	\includegraphics[width=15.0cm]{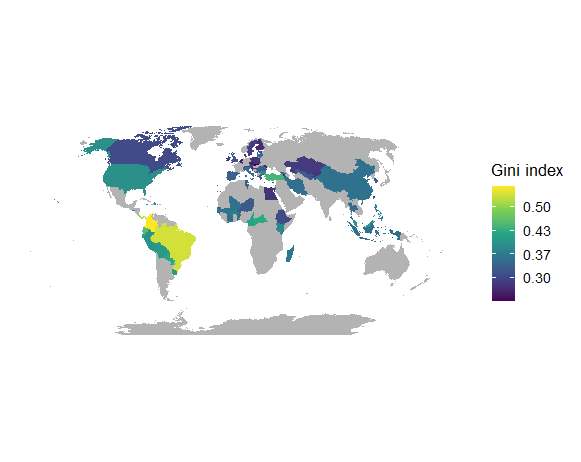}
	\caption{World map of the Gini index for the year 2021.}
	\label{fig:gini-map}
\end{figure}

It is important to emphasize that this application is purely illustrative. The objective is not to provide a substantive analysis of global income inequality patterns, but rather to demonstrate how the proposed models behave when applied to real data supported on the unit interval.

\subsection{Fitting unit-interval models}

In this stage, the models $W$ and $Z$ introduced in Section~\ref{sec:02} were fitted to the 2021 Gini index data, alongside the Beta and Kumaraswamy distributions, which are commonly used for modeling continuous variables on $(0,1)$. In all cases, parameter estimation was carried out via maximum likelihood. For models $W$ and $Z$, a multi-start strategy was employed for the parameters $(\lambda_1, \lambda_2, r)$, while for the Beta and Kumaraswamy distributions, initial values were obtained using sample moments. This approach enhances numerical stability and enables a fair comparison of the goodness of fit across competing models.

Table~\ref{tab:ajuste_gini2021} summarizes the maximum likelihood estimates, the associated standard errors—computed from the observed information matrix whenever it was numerically invertible—and the Akaike (AIC) and Bayesian (BIC) information criteria. In terms of likelihood, model $Z$ achieves the highest log-likelihood value ($\widehat{\ell} \approx 108.96$), followed closely by the Beta ($\widehat{\ell} \approx 104.90$) and $W$ ($\widehat{\ell} \approx 104.27$) models, whereas the Kumaraswamy distribution shows inferior performance ($\widehat{\ell} \approx 96.29$). When penalization for model complexity is taken into account, the AIC favors model $Z$ (AIC $\approx -207.93$), with the Beta distribution ranking second (AIC $\approx -205.80$), while the $W$ and Kumaraswamy models receive less favorable scores. Conversely, the BIC—which imposes a stronger penalty for the number of parameters—favors the Beta distribution (BIC $\approx -201.09$), reflecting its more parsimonious two-parameter structure compared to the five-parameter formulations of models $W$ and $Z$.

From a numerical standpoint, the asymptotic standard errors of the Beta shape parameters are moderate, indicating good parameter identifiability for this model. For the Kumaraswamy distribution, the estimate of parameter $b$ reaches a large value ($\widehat{b} \approx 100$), close to the upper bound imposed during optimization, and is associated with a large standard error, suggesting difficulties in parameter identification. For models $W$ and $Z$, the computation of the observed information matrix reveals some instability: several standard errors associated with the scale parameters are non-numerical (NaN), indicating strong parameter correlations and relatively flat regions of the log-likelihood surface. Although this does not invalidate the use of information criteria for global model comparison, it calls for caution when interpreting individual uncertainty measures for $\widehat{\lambda}_1$ and $\widehat{\lambda}_2$.

\begin{table}[!htpb]
	\centering
	\caption{Maximum likelihood estimates (standard errors in parentheses, when available) and model selection criteria for the Gini index data in 2021.} 
	\label{tab:ajuste_gini2021}
	\begin{tabular}{lcccc}
		\toprule
		& $W$ & $Z$ & Beta & Kumaraswamy \\ 
		\midrule
		$\alpha_1$ & 44.234 
		(n.e.) & 0.659 
		(0.624) & 19.588 
		(3.113) & 4.734 
		(0.338) \\ 
		$\alpha_2$ & 128.700 
		(n.e.) & 47.615 
		(n.e.) & 36.348 
		(5.810) & 100.000 
		(31.302) \\ 
		$\lambda_1$ & 70.330 
		(3.2e-04) & 3.976 
		(n.e.) & — & — \\ 
		$\lambda_2$& 11.879 
		(5.4e-04) & 0.117 
		(n.e.) & — & — \\ 
		$r$ & 4.544 
		(0.001) & 0.059 
		(0.043) & — & — \\ 
		LogLik & 104.27 & 108.96 & 104.90 & 96.29 \\ 
		AIC & -198.55 & -207.93 & -205.80 & -188.58 \\ 
		BIC & -186.76 & -196.15 & -201.09 & -183.87 \\ 
		\bottomrule
	\end{tabular}
\end{table}

Figure~\ref{fig:hist-gini-fits} complements the results in Table~\ref{tab:ajuste_gini2021} by illustrating the fit of the four distributions to the histogram of the observed Gini coefficients in 2021. The Beta and $Z$ models track the empirical shape more closely, providing a balanced representation of the skewness present in the data.

Model $Z$ attains the highest log-likelihood value and the lowest AIC, confirming its strong overall fit. Visually, the fitted density closely follows the modal behavior of the data, although it slightly overestimates density around the mode and decays more rapidly in the right tail. These discrepancies are subtle and do not compromise the model’s overall performance.

The Beta distribution also yields a consistent fit, particularly by adequately capturing both the central behavior and tail structure. Combined with its parsimonious specification, this explains why it is favored by the BIC, which more heavily penalizes model complexity.

Model $W$ exhibits intermediate performance. Its estimated density reproduces the central region of the histogram reasonably well, but does not achieve the same level of adherence as the Beta and $Z$ models, which is reflected in the information criteria. In the case of the Kumaraswamy distribution, difficulties in simultaneously representing concentration and dispersion become apparent both in the shape of the fitted curve and in the numerical instability of parameter $b$, whose large estimated value limits model flexibility.
\begin{figure}[H]
	\centering
	\includegraphics[width=10.8cm]{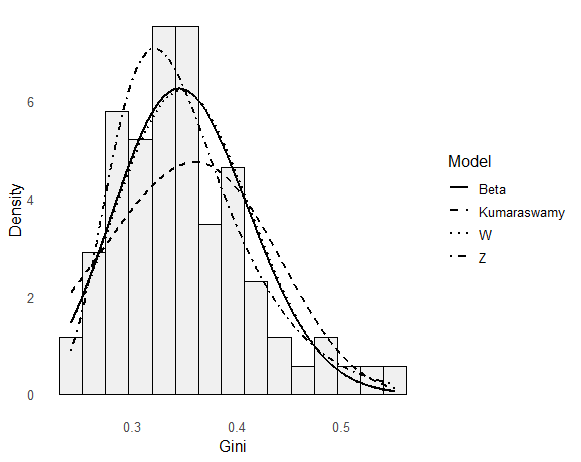}
	\caption{Histogram of Gini coefficients (2021) with fitted density curves for the $W, Z$, Beta, and Kumaraswamy models.}
	\label{fig:hist-gini-fits}
\end{figure}

As a complement to the density-based analysis, Figure~\ref{fig:ecdf-cdf-gini} compares the empirical distribution function of the Gini coefficients with the fitted cumulative distribution functions obtained from each model. The empirical distribution is shown as a step function, while the dashed continuous curves correspond to the fitted CDFs based on maximum likelihood estimates. The bands surrounding the empirical curve represent the expected sampling variability and serve as a visual reference for assessing goodness of fit.

Model $Z$ exhibits a high level of agreement between the empirical and fitted cumulative distributions over almost the entire support. The Beta distribution also shows good overall performance, although with more noticeable discrepancies in intermediate regions. Model $W$ displays more pronounced deviations, particularly in the central part of the distribution, while the Kumaraswamy model shows weaker global adherence, with systematic departures across much of the support. These findings are consistent with the evidence provided by the information criteria and anticipate the patterns observed in the residual diagnostics.
\begin{figure}[H]
	\centering
	\includegraphics[width=15.0cm]{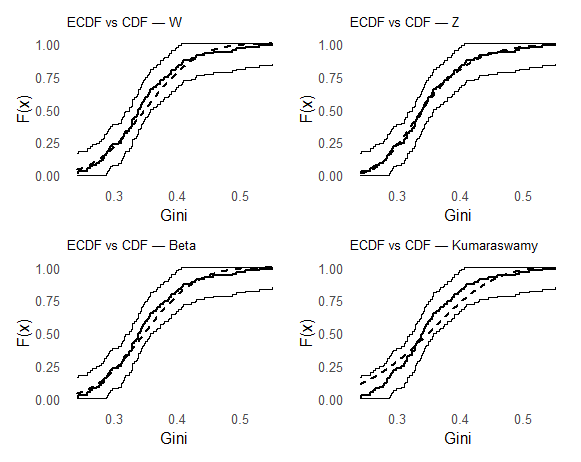}
	\caption{Comparison between the empirical distribution function and the fitted cumulative distribution functions for the $W, Z$, Beta, and Kumaraswamy models applied to the 2021 Gini index data.}
	\label{fig:ecdf-cdf-gini}
\end{figure}

Finally, the analysis of randomized quantile residuals ($R_Q$) and generalized Cox--Snell (GCS) residuals, illustrated in Figures~\ref{fig:qq-rq-gini} and~\ref{fig:qq-gcs-gini}, provides a complementary diagnostic of model adequacy. For the $Z$ and Beta distributions, the residuals are largely aligned with the theoretical reference lines, with more pronounced deviations only in the tails, which is consistent with the presence of countries exhibiting extreme inequality levels. In contrast, the $W$ and Kumaraswamy models show more substantial departures, particularly in the upper quantiles, indicating weaker adherence to the data. These patterns reinforce the conclusions drawn from the information criteria, confirming that models $Z$ and Beta provide superior fits to the 2021 Gini data. Nevertheless, the application demonstrates that models $W$ and $Z$ are competitive with widely used distributions such as the Beta and Kumaraswamy when modeling continuous data on the unit interval.

\begin{figure}[H]
	\centering
	\includegraphics[width=15.0cm]{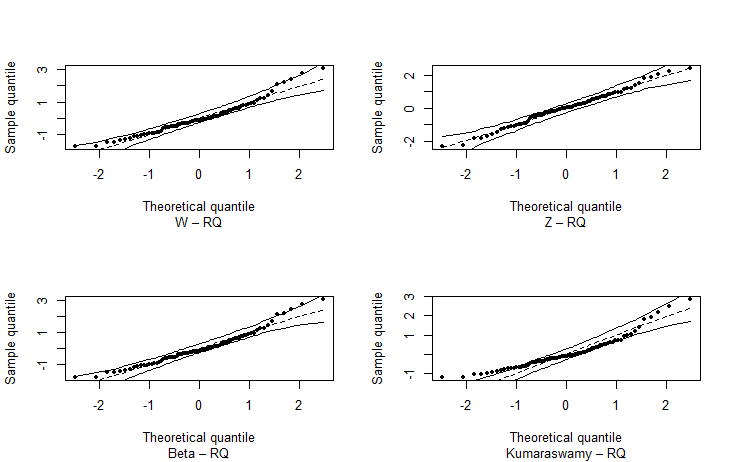}
	\caption{Q--Q plots of randomized quantile residuals ($R_Q$) for the $W, Z$, Beta, and Kumaraswamy models, with simulated envelopes.}
	\label{fig:qq-rq-gini}
\end{figure}

\begin{figure}[H]
	\centering
	\includegraphics[width=15.0cm]{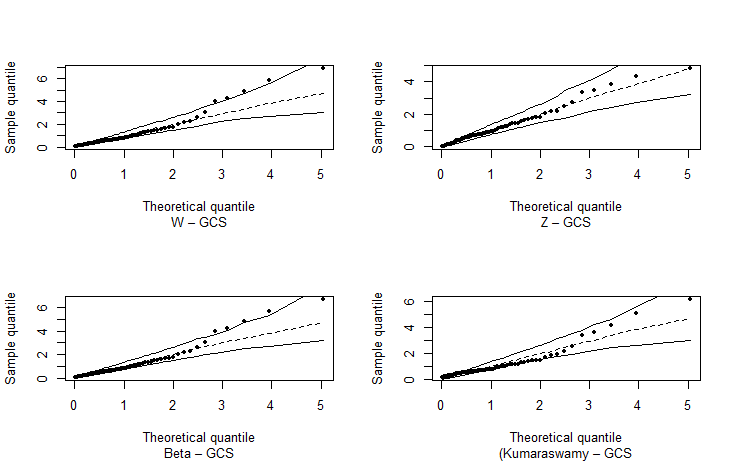}
	\caption{Q--Q plots of generalized Cox--Snell (GCS) residuals for the $W, Z$, Beta, and Kumaraswamy models, with simulated envelopes.}
	\label{fig:qq-gcs-gini}
\end{figure}


\section{Conclusions}\label{sec:08}

This work introduces two novel families of probability distributions on the unit interval, obtained through non-injective transformations of the gamma ratio. The proposed models extend classical unit-interval distributions and establish a direct connection with well-known inequality measures, such as the Gini and Atkinson indices. The derived closed-form expressions for density, cumulative distribution, and raw moments, alongside maximum likelihood estimation procedures, demonstrate both analytical tractability and practical applicability. The flexibility of the models is highlighted through the parameter $r$, which controls the shape of the resulting distributions, allowing the capture of diverse features such as unimodality and skewness. Potential applications span econometrics, reliability analysis, and other fields where modeling proportions or normalized quantities is essential. Future work may explore correlated gamma components, extend the transformations to multivariate settings, and investigate Bayesian estimation procedures for these models.


%
%
	\paragraph*{Acknowledgements}
The research was supported in part by CNPq and CAPES grants from the Brazilian government.
	
	\paragraph*{Disclosure statement}
	There are no conflicts of interest to disclose.



\newpage
\begin{appendices}
	\section{Some technical results}

The proof of the following result is a direct consequence of the gamma distribution definition and is therefore omitted.
\begin{proposition}\label{prop-a-1}
	If $X\sim\text{Gamma}(\alpha_1,\lambda_1)$ and $Y\sim\text{Gamma}(\alpha_2,\lambda_2)$, then, we have:
	\begin{align*}
		\mathbb{E}\left[\exp(-Xx)\right]
		=
		{\lambda_1^{\alpha_1}\over(x+\lambda_1)^{\alpha_1}},
		\quad 
		\mathbb{E}\left[X\exp(-Xx)\right]
		=
		{\alpha_1\lambda_1^{\alpha_1}\over (x+\lambda_1)^{\alpha_1+1}},
	\end{align*}
		and
		\begin{align*}
		\mathbb{E}\left[\exp(-Yx)\right]
		=
		{\lambda_2^{\alpha_2}\over(x+\lambda_2)^{\alpha_2}},
		\quad 
		\mathbb{E}\left[Y\exp(-Yx)\right]
		=
		{\alpha_2\lambda_2^{\alpha_2}\over (x+\lambda_2)^{\alpha_2+1}}.
	\end{align*}
\end{proposition}

\begin{proposition}\label{prop-a-2}
	If $X\sim\text{Gamma}(\alpha_1,\lambda_1)$ and $Y\sim\text{Gamma}(\alpha_2,\lambda_2)$, then, we have:
	\begin{align*}
		\int_{0}^{\infty}
		\left[
		\int_{0}^{u}
		\exp(-vx)
		{\rm d}F_Y(v)
		\right]
		u
		\exp(-ux)
		{\rm d}F_X(u)
		=
		{\lambda_1^{\alpha_1}\lambda_2^{\alpha_2}\over \text{B}(\alpha_1,\alpha_2+1)} \,
		{\,_{2}F_{1}\left(\alpha_1+\alpha_2+1,1;\alpha_2+1;{x+\lambda_2\over 2x+\lambda_1+\lambda_2}\right)\over (2x+\lambda_1+\lambda_2)^{\alpha_1+\alpha_2+1}},
\end{align*}
and
\begin{align*}
		\int_{0}^{\infty}
		\left[
		\int_{0}^{u}
		v\exp(-vx)
		{\rm d}F_Y(v)
		\right]
		\exp(-ux)
		{\rm d}F_X(u)
		&=
		{\lambda_1^{\alpha_1}\lambda_2^{\alpha_2}\over \text{B}(\alpha_1,\alpha_2)} \,
		{\,_{2}F_{1}\left(\alpha_1+\alpha_2,1;\alpha_2+1;{x+\lambda_2\over 2x+\lambda_1+\lambda_2}\right)\over(x+\lambda_2)(2x+\lambda_1+\lambda_2)^{\alpha_1+\alpha_2}}
		\\[0,2cm]
		&-
		{\lambda_1^{\alpha_1}\lambda_2^{\alpha_2}\over\text{B}(\alpha_1,\alpha_2)
		} \,
		{1\over (x+\lambda_2)(2x+\lambda_1+\lambda_2)^{\alpha_1+\alpha_2}},
	\end{align*}
	where ${}_{2}F_{1}(a,b;c;z)$ is the hypergeometric function,  defined for $\vert z\vert<1$ by the power series
	\begin{align}\label{hypergeometric function}
		{}_{2}F_{1}(a,b;c;z)=\sum _{n=0}^{\infty }{\frac {(a)_{n}(b)_{n}}{(c)_{n}}}{\frac {z^{n}}{n!}}.
	\end{align}
	Here $(q)_n$ is the (rising) Pochhammer symbol, which is defined by:
	\begin{align*}
		{\displaystyle (q)_{n}={\begin{cases}1&n=0,
					\\
					q(q+1)\cdots (q+n-1)&n>0.
		\end{cases}}}
	\end{align*}
\end{proposition}
\begin{proof}
	It is simple to observe that
	\begin{align*}
		\int_{0}^{u}
		\exp(-vx)
		{\rm d}F_Y(v)
		=
		{\lambda_2^{\alpha_2}\over \Gamma(\alpha_2) (x+\lambda_2)^{\alpha_2}} \,
		\gamma(\alpha_2,(x+\lambda_2)u)
	\end{align*}
	with $\gamma(s,x)=\int_{0}^{x}t^{s-1}\exp(-t) {\rm d}t$ being the lower incomplete gamma function,
	and
	\begin{align*}
		\int_{0}^{u}
		v\exp(-vx)
		{\rm d}F_Y(v)
		&=
		{\lambda_2^{\alpha_2}\over \Gamma(\alpha_2) (x+\lambda_2)^{\alpha_2+1}} \,
		\,
		\gamma(\alpha_2+1,(x+\lambda_2)u)
		\\[0,2cm]
		&=
		{\lambda_2^{\alpha_2}\over \Gamma(\alpha_2) (x+\lambda_2)^{\alpha_2+1}} \,
		\left[
		\alpha_2
		\gamma(\alpha_2,(x+\lambda_2)u)
		-
		(x+\lambda_2)^{\alpha_2}u^{\alpha_2}
		\exp\{-(x+\lambda_2)u\}
		\right],
	\end{align*}
	where in the last line the well-known recurrence relation: $\gamma(s+1,x)=s\gamma(s,x)-x^{s}\exp(-x)$, have been used.
	Consequently, 
	{
		\begin{multline*}
			\int_{0}^{\infty}
			\left[
			\int_{0}^{u}
			\exp(-vx)
			{\rm d}F_Y(v)
			\right]
			u
			\exp(-ux)
			{\rm d}F_X(u)
			\\[0,2cm]
			=
			{\lambda_1^{\alpha_1}\lambda_2^{\alpha_2}\over \Gamma(\alpha_1)\Gamma(\alpha_2) (x+\lambda_2)^{\alpha_2}} \int_0^\infty
			u^{\alpha_1}
			\exp\{-(x+\lambda_1)u\}
			\gamma(\alpha_2,(x+\lambda_2)u)
			{\rm d}u
		\end{multline*}
	} 
	and
	{
		\begin{align*}
			\int_{0}^{\infty}
			\Bigg[
			&\int_{0}^{u}
			v\exp(-vx)
			{\rm d}F_Y(v)
			\Bigg]
			\exp(-ux)
			{\rm d}F_X(u)
			\\[0,2cm]
			&
			=
			{\alpha_2\lambda_1^{\alpha_1}\lambda_2^{\alpha_2}\over \Gamma(\alpha_1)\Gamma(\alpha_2) (x+\lambda_2)^{\alpha_2+1}} 
			\int_0^\infty
			u^{\alpha_1-1}
			\exp\{-(x+\lambda_1)u\}
			\gamma(\alpha_2,(x+\lambda_2)u)
			{\rm d}u
			\\[0,2cm]
			&
			-
			{\lambda_1^{\alpha_1}\lambda_2^{\alpha_2}\Gamma(\alpha_1+\alpha_2)\over \Gamma(\alpha_1)\Gamma(\alpha_2)
			} \,
			{1\over (x+\lambda_2)(2x+\lambda_1+\lambda_2)^{\alpha_1+\alpha_2}}.
		\end{align*}
	}\noindent
	Finally, by using the identity $\text{B}(a,b)=\Gamma(a)\Gamma(b)/\Gamma(a+b)$ and  the relation \citep{D'Aurizio2016}:
	\begin{align*}
		\int_0^\infty 
		x^{a-1}
		\exp(-sx)\gamma(b,\theta x){\rm d}x
		=
		{\theta^b\Gamma(a+b)\over b(s+\theta)^{a+b}}
		\,_{2}F_{1}\left(a+b,1;b+1;{\theta\over s+\theta}\right),
	\end{align*}
	in the last two identities, the proof follows.
\end{proof}

\begin{proposition}[Technical result]\label{main-prop}
	For $a,b,c,e>0$ such that $ae<2bc$, and for $p+q>1$, we have
	\begin{align*}
		\int_0^\infty {1\over(ax+b)^p(cx+e)^q}{\rm d}x
		=
		{a^{q-1}\over b^{p+q-1}c^q(p+q-1)}
		\, _2F_1\left(q,p+q-1;p+q;1-{ae\over bc}\right),
	\end{align*}
	where $\, _2F_1(a,b;c;z)$ is the hypergeometric function defined in \eqref{hypergeometric function}.
\end{proposition}
\begin{proof}
	Making the change of variable $y=ax/b$, we have
	\begin{align}\label{id-I}
		\int_0^\infty {1\over(ax+b)^p(cx+e)^q}{\rm d}x
		&=
		{a^q\over ab^{p-1} c^q b^q}
		\int_0^\infty {1\over(y+1)^p({bc\over a} y+e)^q}{\rm d}y
		{a^{q-1}\over b^{p+q-1}c^q}
		\int_0^\infty {1\over(y+1)^p(y+{ae\over bc})^q}{\rm d}y.
	\end{align}
	Using the integral representation of hypergeometric functions:
	\begin{align*}
		\, _2F_1(a,b;c;z)
		=
		{\Gamma(c)\over \Gamma(b)\Gamma(c-b)}
		\int_{0}^{\infty}
		t^{-b+c-1}
		(t+1)^{a-c}
		(t-z+1)^{-a}
		{\rm d}t,
		\quad c>b>0, \ -1<z<1,
	\end{align*}
	with $a=q$, $b=p+q-1$, $c=p+q$ and $z=1-ae/cb$,
	expression in \eqref{id-I} becomes
	\begin{align*}
		={a^{q-1}\over b^{p+q-1}c^q}
		{\Gamma(p+q-1)\over \Gamma(p+q)} 
		\, _2F_1\left(q,p+q-1;p+q;1-{ae\over bc}\right),
	\end{align*}
	provided $ae<2bc$. Finally, invoking the well-known identity: $\Gamma(z+1)=z\Gamma(z)$, the proof follows.
\end{proof}

\begin{proposition}[Proposition 3.1 of \cite{Dutta2025}]\label{Moments-truncated}
	Let $X$ be a real-valued random variable 
	and $p>0$ be a real number. For all $\varepsilon\geqslant 0$ and $\delta>0$, the following identity is satisfied:
	\begin{align*}
		\mathbb{E}(X^p\mathds{1}_{\{\varepsilon<X<\delta\}})
		=
		\varepsilon^{p}
		\mathbb{P}(X>\varepsilon)
		-
		\delta^{p}
		\mathbb{P}(X>\delta)
		+
		p
		\int_{\varepsilon}^{\delta}
		u^{p-1}
		\mathbb{P}(X>u)
		{\rm d}u.
	\end{align*}
\end{proposition}

\bigskip 
For simplicity, in the next proposition we adopt the following notation:
\begin{align}\label{def-xi}
	\xi(p)
\equiv 	
\int_{1/2}^1
y^p
\left[1-\int_{1-y}^{y} f_{X\over X+Y}(u){\rm d}u\right]{\rm d}y.
\end{align}
\begin{proposition}\label{prop-app-main}
	If $X$ and $Y$ are independent gamma-distributed with parameter vectors $(\alpha_1,\lambda_1)^\top$ and $(\alpha_2,\lambda_2)^\top$, respectively, such that $\lambda_1=\lambda_2$, then
	\begin{multline*}
	\xi(p)
	=
	{1\over p+1}
	\Bigg\{
		{\mathrm {B}(\alpha_1+p+1,\alpha_2)\over \mathrm {B}(\alpha_1,\alpha_2)}
	\left[
	1-I_{1/2}(\alpha_1+p+1,\alpha_2)
	\right]
	+
		{\mathrm {B}(\alpha_1,\alpha_2+p+1)\over \mathrm {B}(\alpha_1,\alpha_2)} \,
	I_{1/2}(\alpha_1,\alpha_2+p+1)
	-
		\left({1\over 2}\right)^{p+1}
	\Bigg\},
	\end{multline*}
	where $p>\max\{-1,-1-\alpha_1,-1-\alpha_2\}$. In the above, $I_x(a,b)={\rm B}(x;a,b)/{\rm B}(a,b)$ denotes the regularized beta function.
\end{proposition}
\begin{proof}
	Note that $\xi(p)$ in \eqref{def-xi} can be written as
\begin{align*}
	\xi(p)
	=
	\int_{1/2}^1
	y^p
	\left[1-F_{X\over X+Y}(y)+F_{X\over X+Y}(1-y)\right]{\rm d}y.
\end{align*}
As $F_{X\over X+Y}(1-y)=1-F_{Y\over X+Y}(y)$, the above term simplifies to
\begin{align*}
	&=
	\int_{1/2}^1
	y^p
	[1-F_{X\over X+Y}(y)]{\rm d}y
	+
	\int_{1/2}^1
	y^p
	[1-F_{Y\over X+Y}(y)]
	{\rm d}y.
\end{align*}
Applying Proposition \ref{Moments-truncated}, the above expression becomes
\begin{align}
	&
	{1\over p+1}
	\left\{
	\mathbb{E}\left[\left({X\over X+Y}\right)^{p+1}\mathds{1}_{\{{1\over 2}<{X\over X+Y}<1\}}\right]
	-
	\left({1\over 2}\right)^{p+1}
	\mathbb{P}\left({X\over X+Y}>{1\over 2}\right)
	\right.
	\nonumber
	\\[0,2cm]
	&\hspace*{6.cm}
	\left.
	+
	\mathbb{E}\left[\left({Y\over X+Y}\right)^{p+1}\mathds{1}_{\{{1\over 2}<{Y\over X+Y}<1\}}\right]
	-
	\left({1\over 2}\right)^{p+1}
	\mathbb{P}\left({Y\over X+Y}>{1\over 2}\right)
	\right\}.
	\label{id-init}
\end{align}

Since $\lambda_1=\lambda_2$, then, it is clear that
\begin{align}\label{id-init-1}
	\mathbb{E}\left[\left({X\over X+Y}\right)^{p+1}\mathds{1}_{\{{1\over 2}<{X\over X+Y}<1\}}\right]
	=
	\int_{1\over 2}^{1}
	\frac{z^{(\alpha_1+p+1)-1}(1-z)^{\alpha_2-1}}{\mathrm {B}(\alpha_1,\alpha_2)}
	{\rm d}z
	=
	{\mathrm {B}(\alpha_1+p+1,\alpha_2)\over \mathrm {B}(\alpha_1,\alpha_2)}
	\left[
	1-I_{1/2}(\alpha_1+p+1,\alpha_2)
	\right]
\end{align}
and, analogously,
\begin{align}\label{id-init-2}
	\mathbb{E}\left[\left({Y\over X+Y}\right)^{p+1}\mathds{1}_{\{{1\over 2}<{Y\over X+Y}<1\}}\right]
	=
	{\mathrm {B}(\alpha_1,\alpha_2+p+1)\over \mathrm {B}(\alpha_1,\alpha_2)} \,
	I_{1/2}(\alpha_1,\alpha_2+p+1),
\end{align}
where in the last step the well-known identity $I_x(a,b)=1-I_{1-x}(b,a)$ was used.

Furthermore,
\begin{align}\label{id-init-3}
	\mathbb{P}\left({X\over X+Y}>{1\over 2}\right)
	=
	1-I_{1/2}(\alpha_1,\alpha_2),
	\quad
	\mathbb{P}\left({Y\over X+Y}>{1\over 2}\right)
	=
	1-I_{1/2}(\alpha_2,\alpha_1)
	=
	I_{1/2}(\alpha_1,\alpha_2).
\end{align}
By plugging \eqref{id-init-1}, \eqref{id-init-2} and \eqref{id-init-3} in \eqref{id-init}, the proof follows.
\end{proof}

\section{Complete results of the simulation study}
\label{ap:simulacao}

This appendix gathers additional results from the simulation study, with particular
emphasis on the scenarios in which the scale parameters are kept fixed
($\lambda_1=\lambda_2=1$). This configuration makes it possible to assess separately
the impact of the shape parameters $(\alpha_1,\alpha_2)$ and of the transformation
parameter $r$ on estimator performance, isolating the effect of scale. For these
models, the bias and RMSE panels obtained from 500 Monte Carlo replications are
reported for each shape scenario, considering $n\in\{50,100,250,500\}$ and
$r\in\{0{.}5,1,2,4\}$. The results corresponding to the full estimation of the five
parameters $(\alpha_1,\alpha_2,\lambda_1,\lambda_2,r)$ are not reproduced here, since
the main interpretation---based exclusively on the bias and RMSE plots---is fully
presented in the main text.

\subsection{Model \texorpdfstring{$W$}{W} — Estimation with \texorpdfstring{$\lambda_1=\lambda_2=1$}{fixed lambdas}}

This section presents the detailed simulation results for model $W$ under
$\lambda_1=\lambda_2=1$. In this setting, only the shape parameters
$(\alpha_1,\alpha_2)$ and the transformation parameter $r$ are estimated by maximum
likelihood. The analysis focuses on the behavior of the estimators across different
sample sizes and values of $r$, highlighting patterns of bias and variability.

The scenarios considered are
$(\alpha_1,\alpha_2)\in\{(0{.}5,0{.}5),\ (0{.}8,1{.}0),\ (1{.}5,1{.}3),\ (2{.}0,2{.}0)\}.$

\subsubsection{Scenario $(\alpha_1,\alpha_2)=(0{.}5,0{.}5)$}

Figure~\ref{fig:Wfixo-bias-rmse-0505} presents the bias and RMSE plots for the
$(0{.}5,0{.}5)$ scenario. This is a low-shape regime, characterized by a strong
concentration of mass near zero. As a consequence, the estimates exhibit higher
variability, especially for the transformation parameter $r$.

The estimates of $\alpha_1$ and $\alpha_2$ display small biases and RMSE values around
$0{.}09$--$0{.}12$ at $n=50$, with a consistent reduction as $n$ increases. In contrast,
the errors associated with $\widehat{r}$ increase rapidly with $r$, reaching high
levels for $r=2$ and $r=4$ in small samples. Although increasing the sample size
substantially reduces these errors, the estimation of $r$ remains more unstable than
that of the shape parameters.

\begin{figure}[H]
	\centering
	
	\begin{minipage}{0.60\linewidth}
		\centering
		\includegraphics[width=\linewidth]{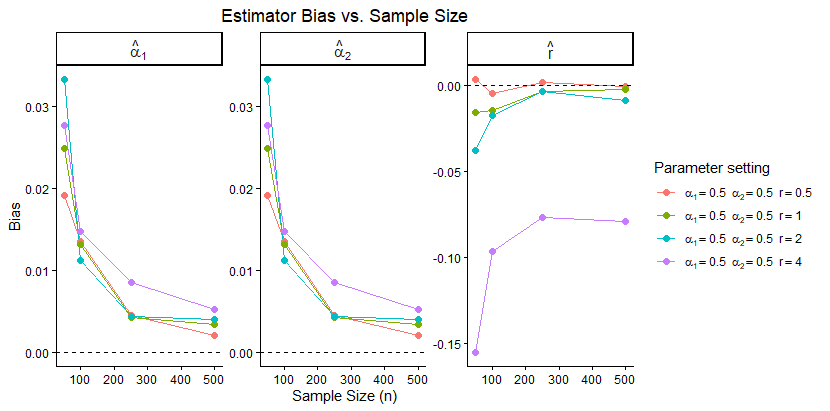}
		\caption*{(a) Bias of the estimators}
	\end{minipage}
	
	\vspace{4mm}
	
	\begin{minipage}{0.60\linewidth}
		\centering
		\includegraphics[width=\linewidth]{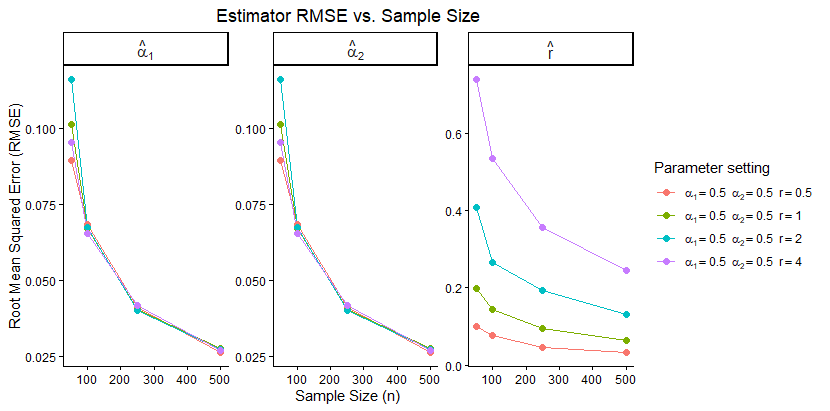}
		\caption*{(b) RMSE of the estimators}
	\end{minipage}
	
	\caption{%
		Bias and RMSE of the estimators for model $W$ with $\lambda_1=\lambda_2=1$
		under the scenario $(\alpha_1,\alpha_2)=(0.5,0.5)$.
	}
	\label{fig:Wfixo-bias-rmse-0505}
\end{figure}

\subsubsection{Scenario $(\alpha_1,\alpha_2)=(0{.}8,1{.}0)$}

The $(0{.}8,1{.}0)$ scenario exhibits moderate asymmetry, leading to improved
stability relative to the $(0{.}5,0{.}5)$ case. The plots in
Figure~\ref{fig:Wfixo-bias-rmse-0810} show a gradual reduction in bias and RMSE for
the shape parameters, especially from $n=100$ onward.

Even so, the parameter $r$ remains highly sensitive in small samples: errors increase
with larger values of $r$, but decrease substantially at $n=250$ and stabilize at
$n=500$. This behavior reinforces that the estimation of $r$ tends to require larger
sample sizes to achieve precision comparable to that obtained for $(\alpha_1,\alpha_2)$.

\begin{figure}[H]
	\centering
	
	\begin{minipage}{0.60\linewidth}
		\centering
		\includegraphics[width=\linewidth]{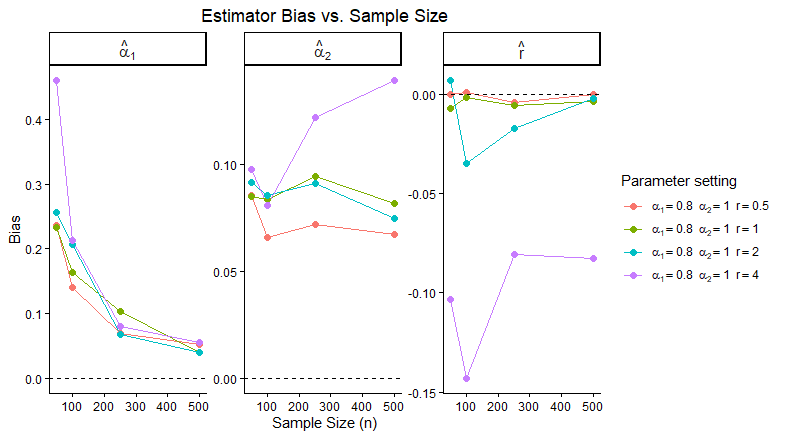}
		\caption*{(a) Bias of the estimators}
	\end{minipage}
	
	\vspace{4mm}
	
	\begin{minipage}{0.60\linewidth}
		\centering
		\includegraphics[width=\linewidth]{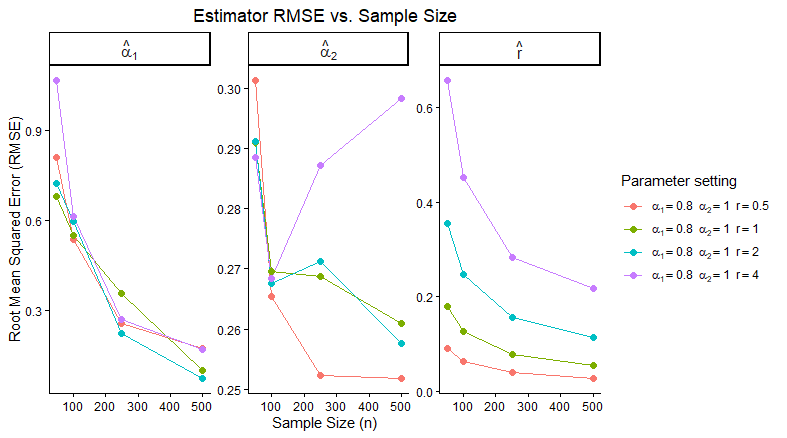}
		\caption*{(b) RMSE of the estimators}
	\end{minipage}
	
	\caption{%
		Bias and RMSE of the estimators for model $W$ with $\lambda_1=\lambda_2=1$
		under the scenario $(\alpha_1,\alpha_2)=(0.8,1.0)$.
	}
	\label{fig:Wfixo-bias-rmse-0810}
\end{figure}

\subsubsection{Scenario $(\alpha_1,\alpha_2)=(1{.}5,1{.}3)$}

The $(1{.}5,1{.}3)$ scenario combines moderate asymmetry with higher shape values,
yielding noticeably more stable estimators. Figure~\ref{fig:Wfixo-bias-rmse-1513}
shows that the shape parameters have small biases and RMSE values that decrease
consistently as $n$ grows.

Despite the overall improvement, $\widehat{r}$ still concentrates the largest
variability in the model, particularly for $r=4$. The pattern, however, is smoother
than in low-shape scenarios, indicating that larger values of $(\alpha_1,\alpha_2)$
help stabilize the estimation procedure.

\begin{figure}[H]
	\centering
	
	\begin{minipage}{0.60\linewidth}
		\centering
		\includegraphics[width=\linewidth]{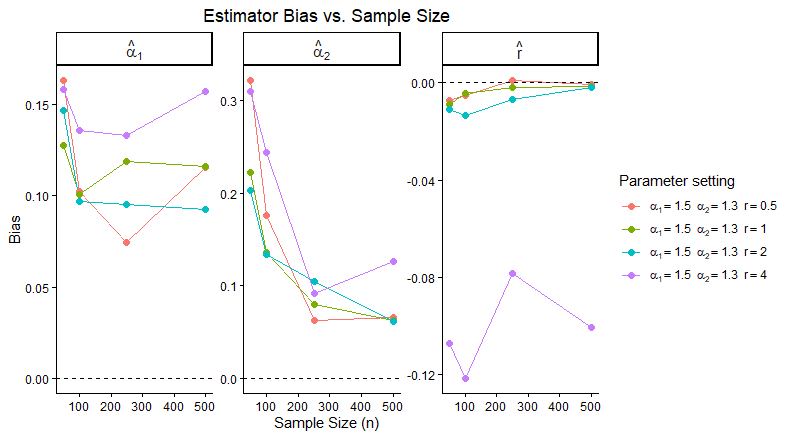}
		\caption*{(a) Bias of the estimators}
	\end{minipage}
	
	\vspace{4mm}
	
	\begin{minipage}{0.60\linewidth}
		\centering
		\includegraphics[width=\linewidth]{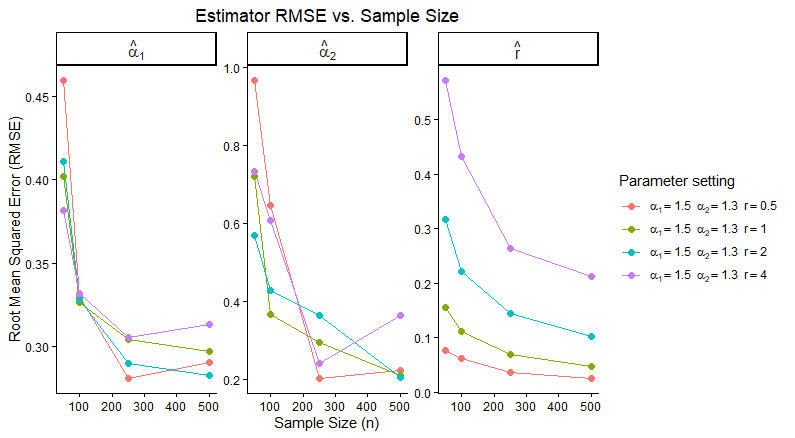}
		\caption*{(b) RMSE of the estimators}
	\end{minipage}
	
	\caption{%
		Bias and RMSE of the estimators for model $W$ with $\lambda_1=\lambda_2=1$
		under the scenario $(\alpha_1,\alpha_2)=(1.5,1.3)$.
	}
	\label{fig:Wfixo-bias-rmse-1513}
\end{figure}

\subsubsection{Scenario $(\alpha_1,\alpha_2)=(2{.}0,2{.}0)$}

In the symmetric $(2{.}0,2{.}0)$ scenario, the most stable behavior among the four
cases is observed. Figure~\ref{fig:Wfixo-bias-rmse-2020} shows that the shape
parameters exhibit small biases and RMSE values that decrease rapidly as $n$
increases. Symmetry contributes to the regularity of the estimates and to the
homogeneity of errors between $\widehat{\alpha}_1$ and $\widehat{\alpha}_2$.

Although $\widehat{r}$ still shows sensitivity to larger values of $r$, its
variability decreases more markedly than in the previous scenarios, reinforcing the
stabilizing effect of symmetry and higher shape values.

\begin{figure}[H]
	\centering
	
	\begin{minipage}{0.60\linewidth}
		\centering
		\includegraphics[width=\linewidth]{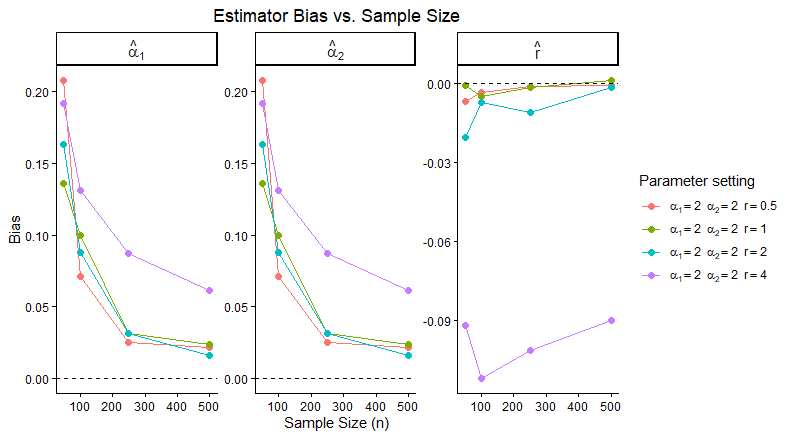}
		\caption*{(a) Bias of the estimators}
	\end{minipage}
	
	\vspace{4mm}
	
	\begin{minipage}{0.60\linewidth}
		\centering
		\includegraphics[width=\linewidth]{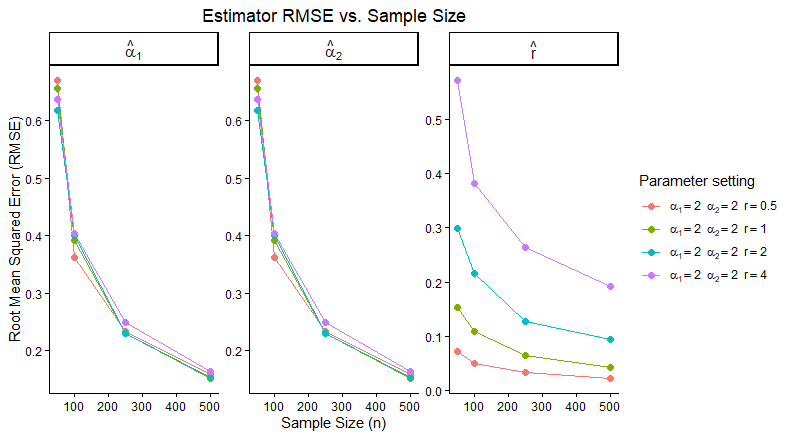}
		\caption*{(b) RMSE of the estimators}
	\end{minipage}
	
	\caption{%
		Bias and RMSE of the estimators for model $W$ with $\lambda_1=\lambda_2=1$
		under the scenario $(\alpha_1,\alpha_2)=(2.0,2.0)$.
	}
	\label{fig:Wfixo-bias-rmse-2020}
\end{figure}

\subsection{Model \texorpdfstring{$Z$}{Z} — Estimation with \texorpdfstring{$\lambda_1=\lambda_2=1$}{fixed lambdas}}

This appendix presents the complete results of the simulation study for model $Z$
under the configuration in which the scale parameters are kept fixed, with
$\lambda_1 = \lambda_2 = 1$. Unlike the full estimation approach, in this
configuration only the shape parameters $(\alpha_1,\alpha_2)$ and the transformation
parameter $r$ are estimated.

The results show that, in contrast to model $W$, fixing the scale parameters is not
sufficient to stabilize the estimators for model $Z$. This instability manifests
through large biases and high RMSE values, revealing important challenges in
parameter identifiability even at large sample sizes.

\subsubsection{Scenario \texorpdfstring{$(\alpha_1,\alpha_2)=(0{.}5,0{.}5)$}{(0.5,0.5)}}

In the symmetric low-shape scenario $(\alpha_1,\alpha_2)=(0{.}5,0{.}5)$, the
estimators $\widehat{\alpha}_1$ and $\widehat{\alpha}_2$ exhibit strong
overestimation for all combinations of $n$ and $r$. At $n=50$, biases exceed $15$ for
$r\leq 2$, with RMSE values above $50$. Although there is a progressive reduction as
$n$ increases, even at $n=500$ the RMSE remains between $7$ and $30$ for $r\leq 2$,
indicating substantial identifiability difficulties. Only for $r=4$ is a marked
improvement observed, with nearly zero bias and RMSE around $0{.}2$ at $n=500$.

The transformation parameter $\widehat{r}$ is even more unstable. For
$r\in\{0{.}5,1,2\}$, the biases at $n=50$ range from $21{.}9$ to $81{.}4$, decreasing
slowly to values between $2{.}8$ and $11{.}7$ at $n=500$. The RMSE remains very large
across the entire sample-size range (above $40$ for $r\leq 2$). Only for $r=4$ is
satisfactory convergence observed, with RMSE close to $1{.}8$ at $n=500$.

This scenario is one of the most challenging regimes for model $Z$, due to the
strong concentration of probability mass near the boundaries when
$\alpha_1=\alpha_2<1$.

\begin{figure}[H]
	\centering
	
	\begin{minipage}{0.60\linewidth}
		\centering
		\includegraphics[width=\linewidth]{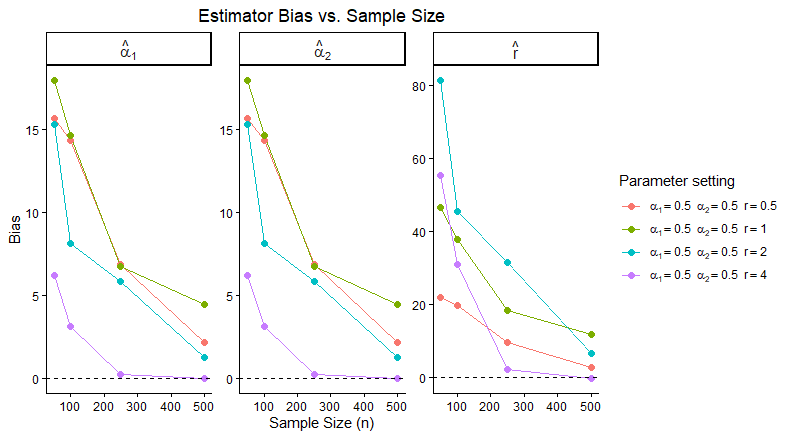}
		\caption*{(a) Bias of the estimators}
	\end{minipage}
	
	\vspace{4mm}
	
	\begin{minipage}{0.60\linewidth}
		\centering
		\includegraphics[width=\linewidth]{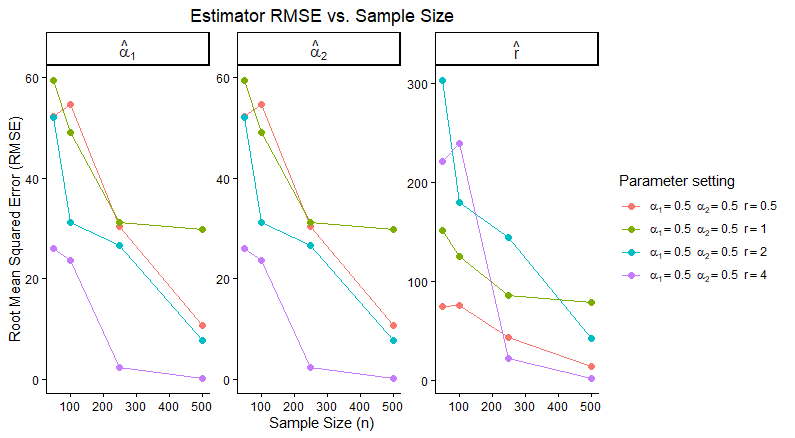}
		\caption*{(b) RMSE of the estimators}
	\end{minipage}
	
	\caption{%
		Bias and RMSE of the estimators for model $Z$ with $\lambda_1=\lambda_2=1$
		under the scenario $(\alpha_1,\alpha_2)=(0.5,0.5)$.
	}
	\label{fig:Zfixo-bias-rmse-0505}
\end{figure}

\subsubsection{Scenario \texorpdfstring{$(\alpha_1,\alpha_2)=(0{.}8,1{.}0)$}{(0.8,1.0)}}

In the moderately asymmetric scenario, the shape estimators remain overestimated for
all combinations of $n$ and $r$. At $n=50$, biases range between $12$ and $21$, with
RMSE values exceeding $45$. Although errors decrease gradually as $n$ increases,
biases between $3$ and $7$ and RMSE values above $10$ are still observed at $n=500$.
Asymmetry has a clear impact: larger values of $r$ amplify biases and RMSE values,
especially for $\widehat{\alpha}_1$.

The estimation of $r$ in this scenario is highly unstable for all true parameter
values. At $n=50$, biases range from about $7$ ($r=0{.}5$) to more than $65$ ($r=4$),
with RMSE between $21$ and $177$. Increasing $n$ partially reduces errors, but they
remain high even at $n=500$, especially for $r=2$ and $r=4$, whose RMSE values remain
close to $60$. Thus, the identification of $r$ remains quite problematic in this
scenario.

\begin{figure}[H]
	\centering
	
	\begin{minipage}{0.60\linewidth}
		\centering
		\includegraphics[width=\linewidth]{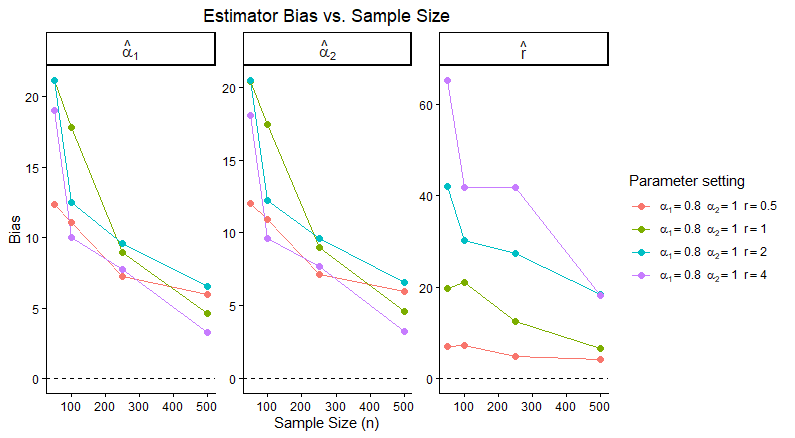}
		\caption*{(a) Bias of the estimators}
	\end{minipage}
	
	\vspace{4mm}
	
	\begin{minipage}{0.60\linewidth}
		\centering
		\includegraphics[width=\linewidth]{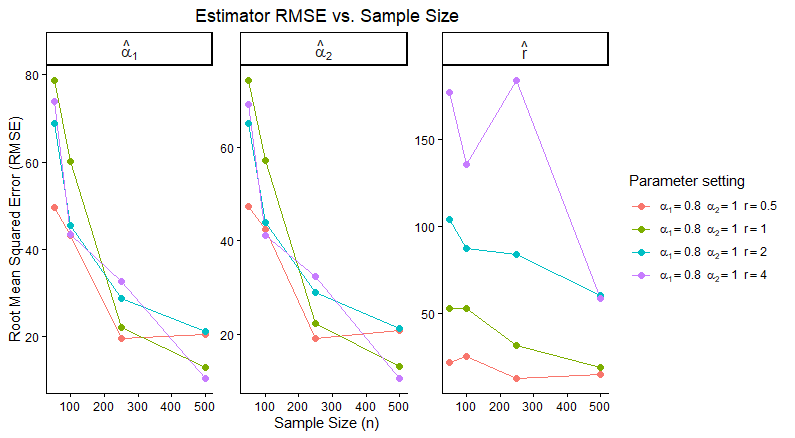}
		\caption*{(b) RMSE of the estimators}
	\end{minipage}
	
	\caption{%
		Bias and RMSE of the estimators for model $Z$ with $\lambda_1=\lambda_2=1$
		under the scenario $(\alpha_1,\alpha_2)=(0.8,1.0)$.
	}
	\label{fig:Zfixo-bias-rmse-0810}
\end{figure}

\subsubsection{Scenario \texorpdfstring{$(\alpha_1,\alpha_2)=(1{.}5,1{.}3)$}{(1.5,1.3)}}

In contrast to the previous two scenarios, the $(1{.}5,1{.}3)$ case exhibits much
more favorable performance. The shape estimators show moderate biases that decrease
as $n$ increases. At $n=50$, the bias of $\widehat{\alpha}_1$ ranges between
$0{.}13$ and $0{.}16$, whereas the bias of $\widehat{\alpha}_2$ lies between
$0{.}20$ and $0{.}32$. The RMSE values decrease from approximately $0{.}4$--$1{.}0$ to
ranges between $0{.}2$ and $0{.}4$ when $n$ reaches $500$, indicating adequate
precision and little impact of the value of $r$.

The transformation parameter is recovered very well for $r\leq 2$, with biases close
to zero and RMSE below $0{.}10$ for $n\geq 250$. For $r=4$, although moderate
negative biases persist, the RMSE values remain relatively small (between
$0{.}19$ and $0{.}21$ at $n=500$). This scenario shows that model $Z$ can be estimated
with good stability when the generating distribution has moderate shape and is not
excessively asymmetric.

\begin{figure}[H]
	\centering
	
	\begin{minipage}{0.60\linewidth}
		\centering
		\includegraphics[width=\linewidth]{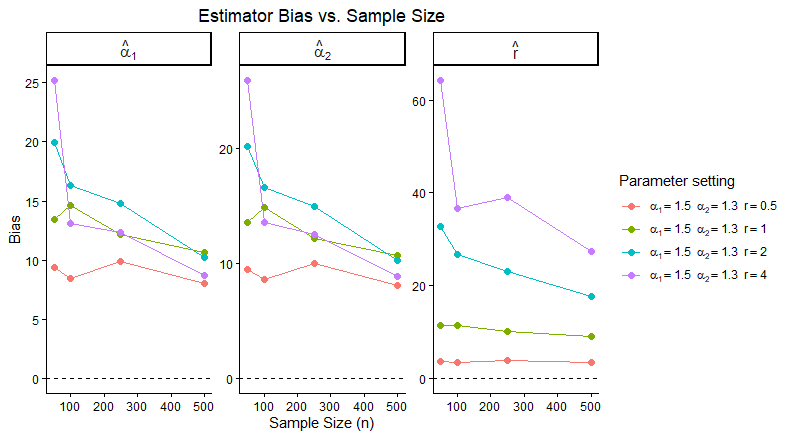}
		\caption*{(a) Bias of the estimators}
	\end{minipage}
	
	\vspace{4mm}
	
	\begin{minipage}{0.60\linewidth}
		\centering
		\includegraphics[width=\linewidth]{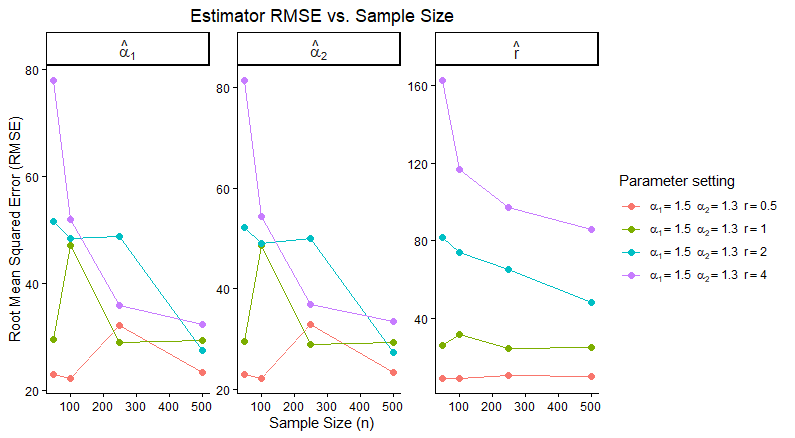}
		\caption*{(b) RMSE of the estimators}
	\end{minipage}
	
	\caption{%
		Bias and RMSE of the estimators for model $Z$ with $\lambda_1=\lambda_2=1$
		under the scenario $(\alpha_1,\alpha_2)=(1.5,1.3)$.
	}
	\label{fig:Zfixo-bias-rmse-1513}
\end{figure}

\subsubsection{Scenario \texorpdfstring{$(\alpha_1,\alpha_2)=(2{.}0,2{.}0)$}{(2.0,2.0)}}

Despite symmetry and higher shape, the $(2{.}0,2{.}0)$ scenario shows unsatisfactory
performance for model $Z$ with fixed scales. The shape estimators remain strongly
overestimated across all sample sizes, with biases ranging between $20$ and $40$ and
RMSE values often above $70$, even at $n=500$. Increases in $n$ lead only to modest
reductions.

The estimation of $r$ is also problematic: for $r\leq 2$, RMSE values range from $13$
to $98$, and for $r=4$ they reach values close to $200$, indicating severe
instability. These results suggest that the combination of larger shape values and
the transformation $T_r(x)$ produces densities that contain relatively little
information for estimating the parameters of model $Z$, even with $\lambda_1$ and
$\lambda_2$ fixed.

\begin{figure}[H]
	\centering
	
	\begin{minipage}{0.60\linewidth}
		\centering
		\includegraphics[width=\linewidth]{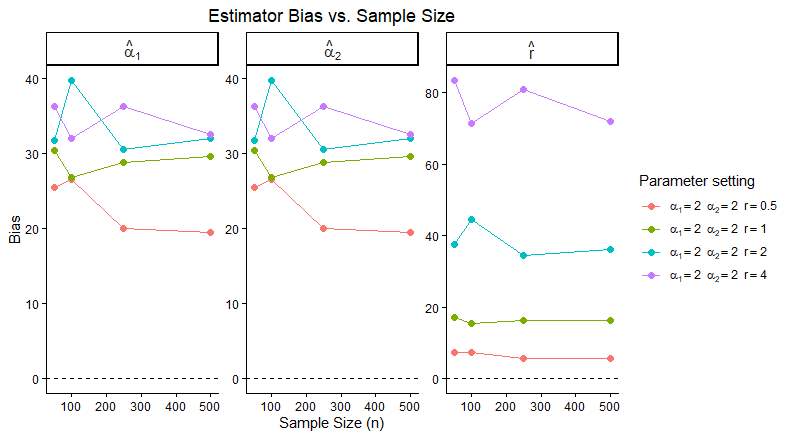}
		\caption*{(a) Bias of the estimators}
	\end{minipage}
	
	\vspace{4mm}
	
	\begin{minipage}{0.60\linewidth}
		\centering
		\includegraphics[width=\linewidth]{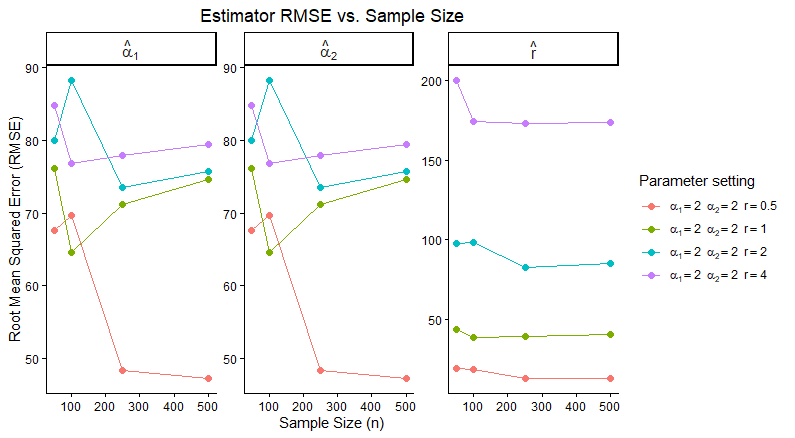}
		\caption*{(b) RMSE of the estimators}
	\end{minipage}
	
	\caption{%
		Bias and RMSE of the estimators for model $Z$ with $\lambda_1=\lambda_2=1$
		under the scenario $(\alpha_1,\alpha_2)=(2.0,2.0)$.
	}
	\label{fig:Zfixo-bias-rmse-2020}
\end{figure}


Unlike what is observed for model $W$, fixing the scale parameters is not sufficient
to stabilize estimation in model $Z$, especially in low-shape or asymmetric
scenarios. In these cases, errors remain large even at $n=500$, particularly for
the transformation parameter. Only in the moderate-shape scenario $(1{.}5,1{.}3)$ is
satisfactory performance observed. These results reinforce the need for larger
sample sizes or alternative estimation methods for model $Z$.

\end{appendices}




\begin{thebibliography}{}
	
	

\bibitem[Abramowitz and Stegun, 1972]{Abramowitz1972}
Abramowitz, M., Stegun, I.A., editors: Handbook of Mathematical Functions with Formulas, Graphs, and Mathematical Tables. Dover, New York, 9th printing edition (1972)

\bibitem[Atkinson, 1970]{Atkinson1970}
Atkinson, A.B.: On the measurement of inequality. J. Econ. Theory 2(3), 244--263 (1970)

  \bibitem[Casella and Berger, 2002]{casella2002statistical}
Casella, G., Berger, R.L.: Statistical Inference. 2nd ed. Duxbury Press, Pacific Grove (2002)

  \bibitem[Colin and Ossikovski, 2025]{colin2025}
Colin, E., Ossikovski, R.: Distribution and moments of a normalized dissimilarity ratio for two correlated gamma variables. Preprint. https://arxiv.org/pdf/2503.08808 (2025)
  
 \bibitem[Cribari-Neto and Zeileis, 2010]{Cribari2010}
Cribari-Neto, F., Zeileis, A.: Beta regression in R. J. Stat. Softw. 34(2), 1--24 (2010)

 \bibitem[Cross et~al., 2006]{crossetal:2006}
Cross, N., Herman, R., Gaunt, C.: Investigating the usefulness of the beta pdf to describe parameters in reliability analyses. In: 2006 International Conference on Probabilistic Methods Applied to Power Systems, pp. 1--6 (2006) 
  
\bibitem[D'Aurizio, 2016]{D'Aurizio2016}
D'Aurizio, J.: Integrating the lower incomplete gamma $\int_0^\infty x^{a-1}e^{-s x}\gamma(b,x)\mathrm{d}x$. Mathematics Stack Exchange. https://math.stackexchange.com/q/1825107 (2016)

\bibitem[Dutta et~al., 2025]{Dutta2025}
Dutta, S., Vila, R., Ribeiro, T.K.A.: A novel approach to generate distributions. Preprint. https://arxiv.org/abs/2508.11861 (2025)

\bibitem[Ferrari and Cribari‐Neto, 2004]{fcn:04}
Ferrari, S., Cribari‐Neto, F.: Beta regression for modelling rates and proportions. J. Appl. Stat. 31(7), 799--815 (2004)

\bibitem[Folland, 1999]{Folland1999}
Folland, G.B.: Real Analysis: Modern Techniques and Their Applications. John Wiley \& Sons, Inc., New York, New York (1999)

\bibitem[Geissinger et~al., 2022]{geissingeretal:2022}
Geissinger, E.A., Khoo, C.L.L., Richmond, I.C., Faulkner, S.J.M., Schneider, D.C.: A case for beta regression in the natural sciences. Ecosphere 13(2), e3940 (2022)

\bibitem[Gentle, 2009]{gentle2009computational}
Gentle, J.E.: Computational Statistics. Springer (2009)

\bibitem[Gini, 1936]{Gini1936}
Gini, C.: On the Measure of Concentration with Special Reference to Income and Statistics. Number 208 in Colorado College Publication, General Series (1936)

\bibitem[Hasell et~al., 2023]{owid-economic-inequality}
Hasell, J., Rohenkohl, B., Arriagada, P., Ortiz-Ospina, E., Roser, M.: Economic inequality. Our World in Data. https://ourworldindata.org/economic-inequality (2023). Accessed 20 Nov 2025

\bibitem[Kumaraswamy, 1980]{kuma1980}
Kumaraswamy, P.: A generalized probability density function for double‐bounded random processes. J. Hydrol. 46, 79--88 (1980)

\bibitem[McDonald and Jensen, 1979]{McDonald1979}
McDonald, J.B., Jensen, B.C.: An analysis of some properties of alternative measures of income inequality based on the gamma distribution function. J. Am. Stat. Assoc. 74, 856--860 (1979)

\bibitem[Vila et~al., 2024a]{Vila2024a}
Vila, R., Balakrishnan, N., Bourguignon, M.: On the distribution of a random variable involved in an independent ratio. Commun. Stat. Theory Methods 54(5), 1427--1440 (2024)

\bibitem[Vila et~al., 2024b]{Vila03072024}
Vila, R., Balakrishnan, N., Saulo, H., Zörnig, P.: Family of bivariate distributions on the unit square: theoretical properties and applications. J. Appl. Stat. 51(9), 1729--1755 (2024)

\bibitem[Vila et~al., 2024c]{Vila2024b}
Vila, R., Balakrishnan, N., Saulo, H., Zörnig, P.: Unit-log-symmetric models: characterization, statistical properties and their applications to analyzing an internet access data. Qual. Quant. 58, 4779--4806 (2024)

\bibitem[Vila and Quintino, 2026]{Vila-Quintino2025}
Vila, R., Quintino, F.: A novel unit-asymmetric distribution based on correlated Fréchet random variables. Qual. Quant. 60, 1443--1471 (2026)

\bibitem[Vila et~al., 2025]{Vila2025}
Vila, R., Saulo, H., Quintino, F., Zörnig, P.: A new unit-bimodal distribution based on correlated Birnbaum--Saunders random variables. Comput. Appl. Math. 44, 83 (2025)

\bibitem[{Wolfram Research, Inc.}, 2024]{WolframResearch2024}
Wolfram Research, Inc.: Mathematica, Version 14.2. Champaign, IL (2024)

\end{thebibliography}

\end{document}